\documentclass[sigconf]{acmart}
\AtBeginDocument{%
  }

\setcopyright{acmlicensed}
\copyrightyear{2018}
\acmYear{2018}
\acmDOI{XXXXXXX.XXXXXXX}
\acmConference[Conference acronym 'XX]{Make sure to enter the correct
  conference title from your rights confirmation email}{June 03--05,
  2018}{Woodstock, NY}
\acmISBN{978-1-4503-XXXX-X/2018/06}

\definecolor{lipicsYellow}{rgb}{0.99,0.78,0.07}

\newtheorem{observation}{Observation}

\newtheoremstyle{noparentheses}
  {\topsep}   % ABOVESPACE
  {\topsep}   % BELOWSPACE
  {\itshape}  % BODYFONT
  {\parindent}       % INDENT (empty value is the same as 0pt)
  {\scshape} % HEADFONT
  {.}         % HEADPUNCT
  {5pt plus 1pt minus 1pt} % HEADSPACE
  {\thmname{#1} \thmnumber{#2} \thmnote{#3}}% CUSTOM-HEAD-SPEC
\theoremstyle{noparentheses}
\newtheorem*{claim*}{Claim}

\usepackage{multirow}
\usepackage{environ,tabularx}

\usepackage{listings}
\lstdefinelanguage{pseudocode}{
  morekeywords={
    algorithm,method,new,and,not,
    if,then,else,while,do,repeat,until,seq,
    seqdo,return,call,
    for,pardo,foreach,print,output,input,exit,
    break,loop,end,begin,goto,par,global,local,
    read,write,stop,idle,procedure,function,
    throw,catch
  },
  sensitive=true,
  morecomment=[l]{//},
  morestring=[b]",
  morestring=[s]{``}{''},
}
\lstdefinestyle{pseudocode}{
  language=pseudocode,
  basicstyle=\small\rmfamily,
  commentstyle=\upshape\color{black!50},
  keywordstyle=\bfseries\itshape,
  identifierstyle=\itshape,
  stringstyle=\rmfamily,
  columns=fullflexible,
  mathescape,
  literate={<-}{{$\gets$\ }}2,
  numbers=left,
  numberstyle=\scriptsize\sffamily,
}

\usepackage{booktabs}

\usepackage{mathtools}

\usepackage{multicol}

\newcommand{\undiradj}{-}
\DeclareMathOperator{\adj}{\sim}
\DeclareMathOperator{\nadj}{\not\sim}

\newcommand\Para{\mathrm{para\text-}}

\newcommand\Class[1]{%
  \mathchoice%
  {\text{\normalfont\small$\mathrm{#1}$}}%
  {\text{\normalfont\small$\mathrm{#1}$}}%
  {\text{\normalfont$\mathrm{#1}$}}%
  {\text{\normalfont$\mathrm{#1}$}}%
}
\newcommand\PClass{\mathrm{p\text{\normalfont-}}\Class}

\newcommand\EM[2]{\Class{RM}_{\mathrm{#1}}^{#2}}
\newcommand\PEM[2]{\PClass{RM}_{\mathrm{#1}}^{#2}}

\makeatletter
\newcommand{\problemtitle}[1]{\gdef\@problemtitle{#1}}% Store problem title
\newcommand{\probleminput}[1]{\gdef\@probleminput{#1}}% Store problem input
\newcommand{\problemparameter}[1]{\gdef\@problemparameter{#1}}% Store problem parameter
\newcommand{\problemquestion}[1]{\gdef\@problemquestion{#1}}% Store problem question

\NewEnviron{problem}{
  \problemtitle{}\probleminput{}\problemquestion{}% Default input is empty
  \BODY% Parse input
  \par\addvspace{.5\baselineskip}
  \noindent
  \begin{tabularx}{\linewidth}{@{\hspace{\parindent}} l X c}
    \multicolumn{2}{@{\hspace{\parindent}}l}{\@problemtitle} \\% Title
    \textbf{Input:} & \@probleminput \\% Input
    \textbf{Question:} & \@problemquestion% Question
  \end{tabularx}
  \par\addvspace{.5\baselineskip}
}

\NewEnviron{parameterizedproblem}{
  \problemtitle{}\probleminput{}\problemparameter{}\problemquestion{}% Default input is empty
  \BODY% Parse input
  \par\addvspace{.5\baselineskip}
  \noindent
  \begin{tabularx}{\linewidth}{@{\hspace{\parindent}} l X c}
    \multicolumn{2}{@{\hspace{\parindent}}l}{\@problemtitle} \\% Title
    \textbf{Input:} & \@probleminput \\% Input
    \textbf{Parameter:} & \@problemparameter \\% Parameter
    \textbf{Question:} & \@problemquestion% Question
  \end{tabularx}
  \par\addvspace{.5\baselineskip}
}
\makeatother

\newcommand{\Lang}[1]{\text{\normalfont\textsc{#1}}}
\newcommand{\PLang}[2][]{\mathrm{p}_{#1}\Lang{-#2}}

\usepackage{xfrac}
\usepackage{tikz}
\usetikzlibrary{graphs, arrows.meta, positioning, backgrounds, shapes, decorations.pathmorphing,}

\definecolor{ba.yellow}{RGB}{252,190,18}
\definecolor{ba.gray}{RGB}{153,153,156}
\colorlet{ba.blue}{blue!80!black}
\colorlet{ba.red}{red!90!black}
\definecolor{ba.orange}{RGB}{233,116,81}
\definecolor{ba.pine}{RGB}{67,154,134}
\colorlet{ba.green}{green!50!black}
\definecolor{ba.violet}{RGB}{88, 53, 94}

\definecolor{Bittersweet}{HTML}{C04F17}
\definecolor{VioletRed}{HTML}{EF58A0}

\tikzset{
  every picture/.style={semithick},%
  > = {Stealth[round,sep]},
  grayed/.style = { black!30 },
  node/.style={
    draw,
    circle,
    minimum size=3mm,
    inner sep=0.3pt,
    font=\scriptsize,
    fill=white
  },%
  vertex/.style={ % Florian's favorite alias...
    node 
  },
  small node/.style={
    node,
    minimum size=4.5pt,
    inner sep=0pt,
    outer sep=0pt,
    font=\tiny
  },%
  on/.style={fill=white,inner sep=-.4pt,circle},
  mapsto/.style={|[sep]->,blue!50!black},
  arbitrary/.style={dashed,draw=black!50},
  present/.style={black},
  missing/.style={black!10},
  virtual node/.style={node,draw=none,fill=none},
  white node/.style={node},
  black node/.style={node, fill=black, text=white},
  gray node/.style={node, arbitrary, fill=black!20},
  red node/.style={node, red!70!black, text=white},
  blue node/.style={node, blue!90!black, text=white},
  green node/.style={node, green!60!black, text=white},
  hilight/.style={red!80!black},
  new/.style={
    line width=.4pt,
    double distance=.8pt,
    draw=white,
    double=lipicsYellow,
  },
  new edge/.style={
    draw=lipicsYellow,
    thick
  },
  modify edge/.style={
    draw=lipicsYellow,
    ultra thick
  },
  delete edge/.style={
    modify edge
  },
  add edge/.style={
    modify edge,
    dotted
  },
  green edge node/.style={
    small node,
    pos=0.7,
    green!60!black,
  },
  red edge node/.style={
    small node,
    pos=0.3,
    red!70!black,
  },
  gate node/.style={
    small node,
    minimum size=7pt,
  },
  gate output/.style={
    double=lipicsYellow,
    line width=.5pt,
    double distance=.8pt,
    draw=white,
  },
  gate input/.style={
    double=Bittersweet,
    line width=.5pt,
    double distance=.8pt,
    draw=white,
  },
  circuit input/.style={
    minimum size=3mm,
    inner sep=0.5pt,
    font=\tiny,
  },
  circuit gate/.style={
    draw,
    minimum size=5mm,
    inner sep=0.5pt,
    font=\footnotesize,
    fill=white
  },
  circuit wire/.style={
    draw,
    -{Stealth[round,sep,length=2mm]},
    thick
  }
}

\def\AtMost{\textcolor{black!50}{p \rlap{\,$\preceq$}}}
\def\AtLeast{\textcolor{black!50}{\llap{$\preceq$\,} p\rlap.}}
\def\Between{\textcolor{black!50}{\llap{$\preceq$\,} p\rlap{\,$\preceq$}}}

\begin{document}

\title{The Descriptive~Complexity of Relation~Modification~Problems}


\author{Florian Chudigiewitsch}
\email{fch@tcs.uni-luebeck.de}
\orcid{0000-0003-3237-1650}
\affiliation{%
  \institution{Institute for Theoretical Computer Science, Universität zu Lübeck}
  \city{Lübeck}
  \country{Germany}
}
\author{Marlene Gründel}
\email{mgruendel@ac.tuwien.ac.at}
\orcid{0000-0003-3470-1326}
\affiliation{%
  \institution{Algorithms and Complexity Group, TU Wien}
  \city{Wien}
  \country{Austria}
}
\author{Christian Komusiewicz}
\email{c.komusiewicz@uni-jena.de}
\orcid{0000-0003-0829-7032}
\affiliation{%
  \institution{Institute of Computer Science, Friedrich Schiller University Jena}
  \city{Jena}
  \country{Germany}
}
\author{Nils Morawietz}
\email{nils.morawietz@uni-jena.de}
\orcid{1234-5678-9012}
\affiliation{%
  \institution{Institute of Computer Science, Friedrich Schiller University Jena}
  \city{Jena}
  \country{Germany}
}
\author{Till Tantau}
\email{tantau@tcs.uni-luebeck.de}
\affiliation{%
  \institution{Institute for Theoretical Computer Science, Universität zu Lübeck}
  \city{Lübeck}
  \country{Germany}
}

\renewcommand{\shortauthors}{Chudigiewitsch et al.}

\begin{abstract}
A relation modification problem gets a logical structure and a natural number~$k$ as input and asks whether $k$ modifications of the structure suffice to make it satisfy a predefined property.
We provide a complete classification of the classical and parameterized
complexity of relation modification problems -- the latter w.\,r.\,t.\ the modification budget~$k$ -- based on the descriptive complexity
of the respective target property. 
We consider different types of logical structures on which
modifications are performed: Whereas monadic structures and undirected graphs
without self-loops each yield their own complexity landscapes, we find that
modifying undirected graphs with self-loops, directed graphs, or arbitrary
logical structures is equally hard w.\,r.\,t.\ quantifier patterns. 

Moreover, we observe that all classes of problems considered in this paper are
subject to a strong dichotomy in the sense that they are either very easy to
solve (that is, they lie in $\Para\Class{AC}^{0\uparrow}$ or $\Class{TC}^0$) or
intractable (that is, they contain $\Class W[2]$-hard or $\Class{NP}$-hard
problems).

\end{abstract}




\received{20 February 2007}
\received[revised]{12 March 2009}
\received[accepted]{5 June 2009}

\maketitle

\section{Introduction}

%== Intro to distance to triviality problems
Approaching computational problems by examining their \emph{distance to
triviality}~\cite{GuoHN04} has greatly enhanced our theoretical understanding of
their inherit complexities and inspired practical algorithm engineering in
recent years. 
To accomplish the latter, one might ask how much (and at which costs) a given input structure needs to be modified so that the resulting structure has a specific (more or less ``trivial'') property that can be algorithmically exploited.
%However diverse such problems may be, the overall goal is to
%modify a given input as little as possible so that the resulting structure has a
%specific property that can be algorithmically exploited. \ck{This sentence gives a wrong impression of distance to triviality. This is an approach for parameterization, not necessarily related to modification problems. An example would be Euclidean TSP parameterized by number~$k$ of cities that are not on the convex hull (problem is trivial if $k=0$). I would think that the common thing about all these modification problems are rather some Maximum parsimony principle/Occam's razor/Max Likelihood-type arguments}
Different variants of
distances between structures are conceivable, depending on the type of modification permitted: For
instance, if the input to a problem is given as a \emph{logical structure}
consisting of a universe and relations of arbitrary arity, we can allow to
\emph{add} tuples to relations (known as the \emph{addition} or \emph{completion
problem}), to \emph{delete} tuples (the \emph{deletion problem}), or allow both
(the \emph{editing problem}). 

%== Edge deletion Problems
Historically, the study of these meta-problems has been dominated by the special
case where the input structure describes a graph and tuples correspond to
\emph{edges}. The interest in such \emph{edge modification
problems}~\cite{Yannakakis78, Yannakakis81} stems from the fact that a variety
of well-studied problems in theoretical computer science can be formulated as
such. For instance, the problem of \emph{cluster deletion} (\emph{cluster
editing}, respectively)~\cite{HuffnerKMN10, KonstantinidisP19} asks whether we
can delete (add or delete, respectively) at most $k$ edges of an input graph to
obtain a cluster graph, that is, a disjoint union of cliques. Similarly, the
problem of \emph{triangle deletion}~\cite{BrugmannKM09, ShengX23} tasks us to
delete at most $k$ edges to make a graph triangle-free. As another example, the
\emph{feedback arc set} problem~\cite{Kudelic22} asks whether deleting at most
$k$ edges suffices to remove all cycles from a given directed graph. There are further
applications of graph modification problems in related disciplines such as
machine learning~\cite{BBC04}, network analysis~\cite{FSM15, XSQ24},
sociology~\cite{LT08, LSB12}, databases~\cite{GS76}, computational
biology~\cite{B59, BSY99, RBIL02}, computational physics~\cite{BPLRK16}, and
engineering~\cite{HH87, CNR89, BETT94, JM96}. Analyzing the underlying
meta-problem of graph modification towards a certain property instead of the aforementioned individual
problems offers the obvious advantage of enabling solution procedures to be
formulated as generally (and thus reusable) as possible.

%== studying complexity
While for the closely related meta-problem of \emph{vertex modification},
polynomial-time solvability has been already precisely characterized for hereditary properties in 1980 by Lewis and Yannakakis~\cite{LewisY80}, analyzing the computational
complexity of edge modification problems proved to be more persistent and has
been the subject of numerous publications in the last decades, many of which can
be found in the surveys~\cite{BurzynBD06, Mancini08, Natanzon2001109}. Since
modification problems are equipped with a natural modification budget $k$, they
are particularly suited to be studied through the lens of \emph{parameterized
complexity}~\cite{DowneyF99, Cygan15}, where one analyzes the running time of
algorithms not only in terms of the input size $|I|$, but also with respect to
a numerical parameter $k$. From this perspective, too, the complexity of graph
modification problems has already been extensively investigated. For an
overview, we refer the reader to the recent survey of Crespelle et
al.~\cite{Crespelle2023100556}.

%== Descriptive Complexity of edge deletion
One way to systematically perform the computational complexity analysis of edge
modification problems is to group problems according to the \emph{descriptive
complexity}~\cite{I98, G17} of the targeted graph class. This approach is
rooted in the hope that if a graph class is easy to describe, then the
corresponding edge modification problem, which attempts to transform a given graph into a graph within that class, may also be easy to
solve. A natural way to measure the descriptive complexity of a property is to
analyze the \emph{quantifier pattern} that is needed to express this property in a
certain logic. Börger et al.\ initiated this line of research by providing a
complete classification of decidable fragments of first-order logic based on
their quantifier pattern~\cite{BorgerGG1997}. Subsequent research pursued
a complexity-theoretic approach and explored to which extent quantifier patterns
of second-order problem formulations affect the
classical~\cite{EiterGG00,GottlobKS04,Tantau15} and parameterized
complexity~\cite{BannachCT23} of these problems.

Most pertinent to our work, Fomin et al.~\cite{FominGT20}
investigated the parameterized complexity of vertex and edge modification
problems based on the number of \emph{quantifier alternations} in the
first-order formalization of the targeted property. With respect to edge
modification, they found that $\Sigma_2$-formulas (first-order formulas starting
with existential quantifiers, followed by universal quantifiers) can only
describe \emph{fixed-parameter tractable} (\emph{FPT}) problems -- that is,
problems that can be solved by algorithms that run in time $f(k)\cdot
|I|^{\mathcal{O}(1)}$, where $f$ is a computable function. On the other hand,
there exist problems expressible as $\Pi_2$-formulas (first-order formulas
starting with universal quantifiers, followed by existential quantifiers) that
are $\Class W[2]$-hard.

\subsubsection*{Our Setting.}

%== Full classification
In this paper, we provide a complete and extensive classification of the
classical and parameterized complexity of edge deletion, edge addition, and edge
editing problems based on the first-order quantifier patterns of the targeted
properties. We go beyond the problem setting proposed in~\cite{FominGT20} by
considering three different types of input graphs: \emph{basic graphs}
(occasionally also called \emph{simple graphs}) which are undirected and have no
self-loops, \emph{undirected graphs,} and \emph{directed graphs}. Since we do
observe differences in the complexity landscapes across these various input
settings, we set out to examine the input-sensitivity of modification problems
more rigorously by allowing arbitrary logical structures as inputs. We refer to
these meta-problems as \emph{relation modification problems}. 

In
Table~\ref{table:summary}, we provide an overview of our main results where the following notations are used: 
For a first-order formula over a vocabulary
$\tau = \{R_1, \dots, R_r\}$, the problem
$\Lang{Relation-Modification}_{\text{arb}}^{\text{del}}$ (abbreviated
$\EM{arb}{\text{del}}$) asks us to tell on input of a logical structure
$\mathcal{A} = (A, R_1^{\mathcal{A}}, \dots, R_r^{\mathcal{A}})$ and a natural
number~$k$ whether there is a tuple $D = (D_1,\dots, D_r)$ with
$D_i \subseteq R_i^{\mathcal{A}}$ for all $i\in \{1, \dots, r\}$ and $\|D\|=\sum_{i=1}^r\left|D_i\right|\leq
k$, whose removal from $\mathcal{A}$ yields a structure $\mathcal{A}'$ for which $\mathcal{A}'\models \phi$ holds.
%For a first-order formula over the vocabulary
%$\tau = \{R_1, \dots, R_r\}$, the problem
%$\Lang{Relation-Modification}_{\text{arb}}^{\text{del}}$ (abbreviated
%$\EM{arb}{\text{del}}$) asks us to tell on input of a logical structure
%$\mathcal{A} = (A, R_1^{\mathcal{A}}, \dots, R_r^{\mathcal{A}})$ and a natural
%number~$k$ whether there is a tuple $D = (D_1,\dots, D_r)$ with
%$D_i \subseteq R_i^{\mathcal{A}}$ for all $i\in \{1, \dots, r\}$ and $\|D\|\leq
%k$, where $\|D\| = \sum_{i=1}^r\left|D_i\right|$ is the number of tuples in the
%sets of $D$, whose removal from $\mathcal{A}$ yields a structure $\mathcal{A}'$,
%such that we have $\mathcal{A}'\models \phi$.\ck{Replace this pretty long sentence by shorter ones?} 
We denote with
$\EM{arb}{\text{add}}$ the related problem that asks whether we can add up to
$k$ tuples such that the new structure satisfies $\phi$. Finally, the problem
$\EM{arb}{\text{edit}}$ allows for the addition and deletion of tuples. For a
modification operation $\otimes\in \{\text{del}, \text{add}, \text{edit}\}$, we
denote with $\EM{dir}{\otimes}$ ($\EM{undir}{\otimes}$, and
$\EM{basic}{\otimes}$, respectively) the restriction of the modification problem
to inputs and target structures that are directed graphs (undirected graphs, and
basic graphs, respectively). Further, let $\EM{mon}{\otimes}$ be the
problem restricted to input structures that contain only monadic
relation symbols. For undirected and basic graphs $G = (V, E)$, where we simply
have $D = (D_1)$ with $D_1\subseteq E$ in the case of deletion (and
$D_1\subseteq V\times V$ in the cases of adding and editing), we use the more
natural measure $\|D\| = \left|\{\{u, v\} \mid (u, v) \in
D_1\}\right|$. We assume throughout that the given first-order formula $\phi$ is
in \emph{prenex normal form}, that is, all quantifiers that occur in $\phi$ are
pushed to the beginning of the formula. The precise sequence of universal and
existential quantifiers at the beginning of $\phi$ is captured by a
\emph{(first-order) quantifier pattern~$p \in \{a, e\}^*$}. For
instance, the formula $\phi = \forall u \forall v \exists y\bigl(u = v \lor u
\adj v \lor (u \adj y\land y\adj v)\big)$ has the quantifier pattern $p=aae$.
For all modification operations $\otimes \in \{\text{del}, \text{add},
\text{edit}\}$ and all structure types $\mathrm{T} \in \{\text{arb}, \text{dir},
\text{undir}, \text{basic}, \text{mon}\}$ we let $\EM{T}{\otimes}(p)$ denote the class
of all problems $\EM{T}{\otimes}(\phi)$ where the input formula~$\phi$ exhibits
the quantifier pattern~$p$. Finally, for all problem variants $\EM{T}{\otimes}$,
we denote with $\PEM{T}{\otimes}$ the parameterized version of this problem that treats the modification budget~$k$ as the parameter.

A complete complexity classification of tuple deletion, addition, and editing
problems for first-order formulas depending on the quantifier pattern $p \in
\{a,e\}^*$ (where $p \preceq q$ means that $p$ is a subsequence of~$q$) is given
in Table~\ref{table:summary}. Note that $\Para\Class{AC}^{0\uparrow} \subseteq
\Para\Class P = \Class{FPT}\subseteq \Class W[2]$ holds and that it is a
standard assumption that both $\Para\Class{AC}^{0\uparrow} \subsetneq \Class{FPT}$ and
$\Class{FPT} \subsetneq \Class W[2]$
are true as well. Let us briefly highlight the key insights from
Table~\ref{table:summary} and summarize the main contributions of this paper:

% For instance, $\Lang{cluster-deletion} \in \EM{basic}{\text{del}}(aaa)$ as the earlier formula~$\phi_{\text{clusters}}$ has three universal quantifiers.

\begin{table*}[tbp]
  \centering
  \caption{Complete Classification of the Classical (above) as well as the Parameterized (below) Complexity of Relation Modification Problems. For the latter, the modification budget is treated as the parameter.}
  \label{table:summary}
  \begin{tabular}{llrcl}
    \toprule
    $\EM{undir}{\otimes}(p)$, $\EM{dir}{\otimes}(p)$, and $\EM{arb}{\otimes}(p)$ 
    & $\subseteq \Class{AC^0}$, when 
    &
    & $\AtMost$
    & $e^*$.
    \\
    & $\not\subseteq \Class{AC^{0}}$ but $\subseteq \Class{TC^{0}}$, when 
    & $a$
    & $\Between$
    & $e^*a$.
    \\
    &  $\cap\ \Class{NP}\text{-hard} \neq \emptyset$, when
    & $aa$ or $ae$
    & $\AtLeast$
    &
    \\
    \midrule
    $\EM{basic}{\otimes}(p)$
    & $\subseteq \Class{AC^0}$, when 
    &
    & $\AtMost$
    & $e^*$ or $a$.
    \\
    & $\not\subseteq \Class{AC^{0}}$ but $\subseteq \Class{TC^{0}}$, when 
    & $ea$ or $ae$ or $aa$
    & $\Between$
    & $e^*a$ or $aa$ or $ae$.
    \\
    & \multirow{2}{*}{$\cap\ \Class{NP}\text{-hard} \neq \emptyset$, when}
    & $aea$ or $aee$ or $aae$ or
    & \multirow{2}{*}{$\AtLeast$}
    &
    \\
    & 
    & $eae$ or $eaa$ or $aaa$
    & 
    &
    \\
    \midrule
    $\EM{mon}{\otimes}(p)$
    & $\subseteq \Class{AC^{0}}$, when
    &
    & $\AtMost$
    & $e^*$.
    \\
    & $\not\subseteq \Class{AC^{0}}$ but $\subseteq \Class{TC^{0}}$, when
    & $a$
    & $\AtLeast$
    &
    \\
    \midrule\midrule
    $\PEM{undir}{\otimes}(p)$, $\PEM{dir}{\otimes}(p)$, and $\PEM{arb}{\otimes}(p)$
    & $\subseteq \Para\Class{AC^{0\uparrow}}$, when 
    &
    & $\AtMost$
    & $e^*a^*$.
    \\
    &  $\cap\ \Class{W}[2]\text{-hard} \neq \emptyset$, when
    & $ae$
    & $\AtLeast$
    &
    \\
    \midrule
    $\PEM{basic}{\otimes}(p)$
    & $\subseteq \Para\Class{AC^{0\uparrow}}$, when 
    &
    & $\AtMost$
    & $e^*a^*$ or $ae$.
    \\
    & \multirow{2}{*}{$\cap\ \Class{W}[2]\text{-hard} \neq \emptyset$, when}
    & $aea$ or $aee$
    & \multirow{2}{*}{$\AtLeast$}
    &
    \\
    & 
    & $aae$ or $eae$
    & 
    &
    \\
    \midrule
    $\PEM{mon}{\otimes}(p)$
    & $\subseteq \Para\Class{AC^{0}}$. 
    &
    & 
    & 
    \\
    \bottomrule
  \end{tabular}
\end{table*}
%== First characterization of ED and to FO formulas in classical complexity
%In particular, this paper provides the first classification of the classical
%complexity of edge deletion and edge modification problems to first-order
%formulas in terms of the exact quantifier patterns. Since our results are stronger than a
%classification based on quantifier alternations, classifications with regards to
%quantifier alternations follow as corollaries.

\subsubsection*{Our Contributions.}
As a first set of contributions, we \textbf{refine} the dichotomy of the
\textbf{parameterized complexity of edge modification problems} provided by
Fomin et~al.~\cite{FominGT20} in multiple ways:
\begin{enumerate}
\item First, we resolve the following question, which is central to
understanding the descriptive features that determine the computational
complexity of edge modification problems: Is it the \textbf{number of
alternations} of quantifiers that solely drives the transition from tractability
to intractability, or is this boundary determined by \textbf{short patterns}
that happen to exhibit a specific number of alternations? As we show, the latter
is indeed the case: All intractability results already hold for very short and
simple patterns. While it was previously known that an intractable problem can
be defined by a formula in $\Pi_2$ (more specifically, one with the quantifier
pattern $aeeee$), we find that already the simplest interesting pattern in
$\Pi_2$, namely $ae$, suffices to establish intractability. 
\item We \textbf{improve existing upper bounds} by showing that all classes of
modification problems that are (fixed-parameter) tractable are actually even in
$\Para\Class{AC}^{0\uparrow}$. From an algorithmic perspective, this implies
that these problems admit efficient \textbf{parallel fixed-parameter
algorithms}. From a logical perspective, this yields an interesting dichotomy in
which all considered problem classes are either very easy to solve or contain
intractable problems.
\item We identify \textbf{new intractable problems} that are of independent
interest, such as $\PLang{Edge-Adding To}$ $\Lang{Radius $r$}$ and
$\PLang{Edge-Editing To Radius $r$}$ for all $r \geq 2$.
\item We contrast the parameterized study of edge modification
towards first-order properties by a systematic investigation of their
\textbf{classical complexities}. Here again, we gain interesting insights by
pursuing our fine-grained pattern-based approach: For instance, with respect to
basic graphs, we find that both, $\Sigma_2$ and $\Pi_2$ contain short tractable
(and non-trivial) fragments, whereas a purely alternation-based study would
label the whole class of formulas in $\Pi_1$ as intractable.
\end{enumerate}
 %== Insights: compare to vertex deletion
Our results complement the complexity landscape of \textbf{vertex
deletion problems} obtained by Bannach et al.~\cite{BannachCT24}. In particular,
we find that the parameterized complexity of edge modification problems is
considerably higher: While for directed graphs, vertex deletion towards all
formulas that exhibit the pattern $e^*a^*e^*$ is tractable, we show that with
respect to edge deletion, there are already intractable problems describable by formulas
with pattern~$ae$.

Going beyond edge modification problems, this paper is the first to
observe that the \textbf{complexity of relation modification problems is
sensitive towards the type of logical structures that are allowed as input and
targeted property}. Both in the parameterized and the classical setting, we
observe a complexity jump when transitioning from basic graphs to the more
general undirected graphs. But interestingly, all further generalizations, that
is, allowing directed graphs or even arbitrary logical structures do not
increase the complexity of the respective problems further. To the other end, if
the vocabulary contains only monadic symbols, the problem is always tractable.

%== Insights: Edge deletion and edge editing same complexity
Lastly and contrary to what one might expect, we show that the \textbf{complexities of
tuple deletion, tuple addition, and tuple editing coincide} for all considered fragments.

\subsubsection*{Further Related Work.} 
Besides the lines of work presented above, the interplay of descriptive and parameterized complexity has attracted further recent interest.
Fomin et al.~\cite{FominGT22}, for example, investigated the parameterized
complexity of graph modification to first-order formulas, with respect to modifications within a certain
\emph{elimination distance}. Moreover, the complexities of a
large family of graph modification operations containing, for example, vertex
and edge deletion, edge contraction, and independent set deletion, were recently
considered by Morelle, Sau, and Thilikos~\cite{MorelleST25} for target classes
that are minor-closed. 
%They present $\Class{FPT}$-algorithms with the same running time as the algorithms for vertex deletion. 

The problems we consider in this paper have also been studied in the context of \emph{hereditary model checking in
first-order logic}. Here, a logical structure is a model of a first-order
formula $\phi$ if the structure and all its substructures are models. Bodirsky
and Guzmán-Pro recently characterized the prefix classes for which hereditary
first-order model checking is tractable~\cite{bodirsky2025}. The complexity of
this problem is in turn connected to the complexity of extensional existential
second-order logic~\cite{bodirsky2025-1}.

Another related field of research is the study of \emph{database
repairs}~\cite{Arenas99}, where both deleting tuples~\cite{CHOMICKI200590} and
adding tuples~\cite{Cali03} have been addressed. Although these works also
consider algorithms and complexity for selected problems, they do not rigorously
examine the modification of tuples in a broader sense. \emph{Consistent query
answering (CQA)} is the related problem, where we do not aim to repair a
database according to some specification, but instead ensure that the answers we
retrieve for each query are consistent with some specifications. The classical
complexity of this problem was investigated in~\cite{STAWORKO20101} and its
parameterized complexity in~\cite{Lopatenko07}. The complexity of checking
whether a repair is valid and of CQA with given integrity constraints is studied
in~\cite{arming16}, and recently, a tetrachotomy of CQA with violated primary
key constraints was provided~\cite{10.1145/3452021.3458334}, where the problem
is either solvable in $\Class{AC}^0$, $\Class{NL}$-complete,
$\Class{P}$-complete or $\Class{coNP}$-complete.

\emph{Parallel Parameterized Algorithms} were first investigated by Cesati and
Ianni~\cite{CesatiI98}. Elberfeld, Stockhusen, and Tantau~\cite{ElberfeldST15}
subsequently provided a general framework for parallel parameterized complexity
classes. As a major result, Chen and Flum showed by a parameterized circuit
lower bound that $\Lang{Clique}$ parameterized by the solution size is not in
$\Para\Class{AC}^0$~\cite{ChenF19}. Subsequent
work~\cite{BannachST15,BannachT16,BannachT18, BannachT20,BannachST23,
BannachCT23} pursued this line of research further, with one of the most
surprising results being that the problem $\Lang{Hitting Set}$ on hypergraphs of
constant edge cardinality is in $\Para\Class{AC}^0$ when parameterized by the
solution size. The uniformity of parameterized circuit classes was thoroughly
studied by Hegeman, Martens, and Laarman~\cite{Hegeman0L25} and practical
implementations of parallel parameterized algorithms on parallel random access
machines were designed in~\cite{Abu-KhzamLSS06, CheethamDRST03}.

\subsubsection*{Organization of this Paper.}

We review basic concepts and termino-logy from Finite Model Theory and
Computational Complexity in Section~\ref{section:background}, and then present
the complexity-theoretic classification of relation modification problems for
basic, undirected, and directed graphs as well as for monadic and arbitrary logical
structures in Section~\ref{section:parameterized} for the parameterized, and in
Section~\ref{section:classical} for the classical setting. Statements whose
proofs are moved to the appendix due to space constraints are marked
with~``$\star$'' in the main text.

% \tcsautomoveaddto{main}{
%   \clearpage
%   \appendix
%   \section{Technical Appendix}
%   In the following, we provide the proofs omitted in the main text. In
%   each case, the claim of the theorem or lemma is stated once more for
%   the reader's convenience. 
% }

\section{Descriptive, Parameterized, and Circuit
 Complexity}
\label{section:background}

\subsubsection*{Terminology from Finite Model Theory.}

%== basics of logic and FMT
We use standard terminology from finite model
theory. For a more thorough introduction than provided in the following, see, for
instance~\cite{EbbinghausF05}. A \emph{relational vocabulary}~$\tau$
(also known as a \emph{signature}) is a set of \emph{relation symbols}
to each of which we assign a positive \emph{arity}, denoted by a
superscript. For example, $\tau = \{P^1, E^2\}$ is a relational
vocabulary with a unary relation symbol $P$ and a binary relation
symbol~$E$. In a \emph{monadic} vocabulary, each relation symbol has arity one.
A \emph{finite $\tau$-structure~$\mathcal A$} consists of a
\emph{finite universe~$A$} and for each of the $r$ relation symbols $R_i \in \tau,
i\in\{1, \dots, r\}$ with
arity~$a_i$ of a relation $R_i^{\mathcal A} \subseteq A^{a_i}$.  We denote the
set of \emph{finite $\tau$-structures} as $\Lang{struc}[\tau]$. We assume the
structures are encoded in a reasonable way, see for example~\cite[page
136]{GradelKLMSVVW07}. We further assume all structures to be ordered.

First-order $\tau$-formulas are defined in the usual way. For a
$\tau$-formula $\phi$ with free variables $x_1,\dots,x_\ell$, a 
structure $\mathcal S \in \Lang{struc}[\tau]$ with universe~$S$, and $s_1,\dots,s_\ell \in S$ we write $\mathcal S \models
\phi[s_1,\dots,s_\ell]$ to denote that $\mathcal S$ is a model of
$\phi$ when the free variables are interpreted as the~$s_i$. If $\phi$
has no free variables, we write $\Lang{models}(\phi)$
for the class of finite models of~$\phi$.
%== decision problems
A \emph{decision problem}~$P$ is a subset of $\Lang{struc}[\tau]$
that is closed under isomorphisms. A formula $\phi$ \emph{describes}~$P$ if $\Lang{models}(\phi) = P$. 

For $\tau$-structures
$\mathcal{A}$ and $\mathcal B$ with universes $A$ and~$B$,
respectively, we say that $\mathcal{A}$ is an \emph{induced
substructure} of $\mathcal{B}$ if $A \subseteq B$ and for all $r$-ary 
$R\in \tau$, we have $R^{\mathcal{A}} = R^{\mathcal{B}} \cap A^r$.
For a set $S\subseteq B$, we denote by $\mathcal{B}\setminus S$ the
substructure induced by $B \setminus S$. 

%== graphs in FMT
We regard \emph{directed} graphs $G = (V, E)$ (which are pairs of a
nonempty vertex set~$V$ and an edge relation $E \subseteq V \times V$)
as logical structures $\mathcal G$ over the vocabulary
$\tau_{\text{digraph}} = \{\adj^2\}$ where $V$ is the universe and
$\adj^{\mathcal G} = E$. An \emph{undirected} graph is a directed
graph that additionally satisfies $\phi_{\text{undirected}} \coloneq
\forall x\forall y (x\adj y \to y\adj x)$, while a \emph{basic} graph
satisfies $\phi_{\text{basic}} \coloneq \forall x\forall y \bigl(x\adj
y \to (y\adj x \land x \neq y)\bigr)$. We will use $u \undiradj v$ as
a shorthand for $\{(u, v), (v, u)\} \subseteq E$.

%== quantifier patterns
For a first-order formula in prenex normal form (meaning all
quantifiers are at the front), we can associate a \emph{quantifier prefix  
pattern} (or \emph{pattern} for short), which are words over the alphabet
$\{e, a\}$.
For example, the formula $\phi_{\text{basic}}$ has the
pattern~$aa$, while the formula $\phi_{\text{degree$\geq$2}} \coloneq
\forall x \exists y_1 \exists y_2 \bigl((x \adj y_1) \land (x\adj y_2)
\land (y_1 \neq y_2)\bigr)$ has the pattern~$aee$. As another example,
the formulas in the class~$\Pi_2$ (which start with a universal
quantifier and have one alternation) are exactly the formulas with a
pattern $p \in \{a\}^* \circ \{e\}^*$, which we write briefly as $p
\in a^*e^*$. We write $p \preceq q$ if $p$ is a subsequence
of~$q$. For instance, $aaee \preceq aeaeae$.

%== Edge deletion variants
When we consider budgets for the modification of logical structures $\mathcal{A}
= (A, R_1^{\mathcal{A}}, \dots, R_r^{\mathcal{A}})$, we have several degrees of
freedom: First, we can specify that  modifications are only allowed with respect
to a single relation, or we can allow to modify all relations of the input
structure. We call the first variant \emph{uniform modification} and the second
variant \emph{unrestricted modification}. In this paper, we consider
unrestricted modification, but address uniform modification in the conclusion.
Second, we can choose whether we count each tuple we add or delete (which is the
more natural measure for directed graphs and arbitrary structures), or only the
number of different sets of elements on which we modify relations (this is the
more natural measure on basic and undirected graphs). For directed graphs and
arbitrary structures, we will use the former measure, whereas for basic and
undirected graphs, we will use the latter. However, we will discuss in the
conclusion that applying the respective opposite counting version yields the
same complexity landscape.

Let $\mathcal{A} = (A, R_1^{\mathcal{A}}, \dots, R_r^{\mathcal{A}})$ be a
logical structure where $R_i$ has arity $a_i$ for all $i \in \{1, \dots, r\}$. A
\emph{modulator}~$S = (S_1,\dots, S_r)$ contains sets $S_i \subseteq A^{a_i}$.
If we have $S_i \subseteq R_i^{\mathcal{A}}$ for all $i\in \{1, \dots, r\}$, we
call $S$ a \emph{deletion modulator}, and if we have $S_i \cap R_i^{\mathcal{A}}
= \emptyset$ for all $i\in \{1, \dots, r\}$, we
call $S$ an \emph{addition modulator}. We write
\begin{align*}
  \|S\| \coloneq \begin{cases}
    \sum_{i=1}^r\left|\bigl\{\{u, v\} \mid (u, v) \in S_i\bigr\}\right| & \parbox[t]{3cm}{$\mathcal{A}$ is an undirected or basic graph,}\\
    \sum_{i=1}^r\left|S_i\right| & \text{otherwise.}
  \end{cases}
\end{align*}

Now, for a structure $\mathcal{A} = (A, R_1^{\mathcal{A}}, \dots,
R_r^{\mathcal{A}})$ and a modulator $S = (S_1, \dots, S_r)$ we define 
\begin{align*}
  \mathcal{A} \triangleright S &\coloneq (A,
R_1^{\mathcal{A}} \mathbin{\triangle} S_{1}, \dots, R_r^{\mathcal{A}} \mathbin{\triangle}
S_{r}).
\end{align*}

To illustrate these notions, let us consider the special case of graphs. For a
basic or undirected graph $G = (V, E)$, the solution for the deletion problem is
simply a set $S \subseteq E$ and for the addition or editing problem it is a set
$S \subseteq V \times V$.

\subsubsection*{Relation modification problems.}

We study a wide variety of relation modification problems in this paper.
As an introductory example, consider the following problem where $\phi$ is a first-order $\tau$-formula:

\begin{problem}
\problemtitle{$\EM{arb}{\mathrm{edit}}(\phi)$}
    \probleminput{ A logical structure $\mathcal{A} = (A, R_1^{\mathcal{A}}, \dots,
    R_r^{\mathcal{A}})$ and a number $k\in\mathbb{N}$.}
    \problemquestion{ Is there an modulator $S$ with $\|S\| \leq k$ such that $\mathcal{A} \triangleright S
    \models \phi$?}
\end{problem}

%== edit, delete, plus
Analogously, we define $\EM{arb}{\mathrm{add}}(\phi)$ to be the question of whether there
is an addition modulator $\|S\| \leq k$ such that $\mathcal{A} \triangleright S\models \phi$, and
$\EM{arb}{\mathrm{del}}(\phi)$ that asks whether there is a
deletion modulator $\|S\| \leq k$ with $\mathcal{A} \triangleright S\models \phi$.

%== Problem variants, defined
We define restrictions on the allowed input and target structure by
considering problems $\EM{dir}{\otimes}(\phi)$, where $\otimes\in
\{\mathrm{add}, \mathrm{del}, \mathrm{edit}\}$ and the vocabulary of $\phi$
consists only of a single binary relation $E$, thus forcing input and target
structures to be directed graphs. Similarly, we denote with $\EM{basic}{\otimes}(\phi) =
\EM{dir}{\otimes}(\phi\land \phi_{\text{basic}}) \cap
\Lang{models}(\phi_{\text{basic}})$ the meta-problem for which the input
and target structures are basic graphs, and $\EM{undir}{\otimes}(\phi) =
\EM{dir}{\otimes}(\phi\land \phi_{\text{undir}}) \cap
\Lang{models}(\phi_{\text{undir}})$ forces input and target structures to be
undirected graphs. For a pattern $p \in \{a,e\}^*$, the class
$\EM{arb}{\otimes}(p)$ contains all problems $\EM{arb}{\otimes}(\phi)$ such that
$\phi$ has pattern~$p$. The classes with the subscripts ``basic'', ``undir'' and
``dir'' are defined similarly.

Fomin et al.~\cite{FominGT20} observed that for each class of formulas with the same quantifier alternation, the complexity of edge deletion and
edge addition coincides on basic graphs. Note that the same holds for tuple
deletion and tuple addition on arbitrary logical structures by the following
observation: For a structure $\mathcal{A}  = (A, R_1^{\mathcal{A}}, \dots,
R_r^{\mathcal{A}})$ let $\overline{\mathcal{A}} = (A, A^{a_1}\setminus R_1^{\mathcal{A}}, \dots,
A^{a_r}\setminus R_r^{\mathcal{A}})$ denote its complement structure, and for a formula
$\phi$, let $\overline{\phi}$ denote the formula obtained from $\phi$ by replacing
every occurrence of $R_i(x_1, \dots, x_{a_i})$ by $\neg R_i(x_1, \dots,
x_{a_i})$ in the setting of arbitrary structures, directed and undirected
graphs, and by replacing every occurrence of $(x \adj y)$ by $\neg (x = y) \land
\neg (x\adj y)$ when we restrict the inputs to basic graphs. Then it is not hard
to verify the following.

\begin{observation}\label{observation:del-add-same}
  For every first-order formula $\phi$ and for every type $\mathrm
  T\in\{\mathrm{arb}, \mathrm{dir},\penalty0 \mathrm{undir},\penalty100
  \mathrm{basic}\}$, $(\mathcal{A}, k)$ is a yes-instance of
   $\EM{T}{\mathrm{add}}(\phi)$ if and only if $(\overline{\mathcal{A}}, k)$ is a
  yes-instance of $\EM{T}{\mathrm{del}}(\overline{\phi})$.
\end{observation}

%Consequently, it suffices to study the problem versions that employ general editing operations (that is, additions and deletions) and the ones that only allow for deletions.
We use this observation in subsequent sections by interchangeably employing relation deletions or relation additions, depending on which version is best suited for a concise presentation.

\subsubsection*{Terminology from (Parameterized) Complexity.}

%== Classical
In order to establish tractability in the classical setting, we will show that the respective
problems are contained in the circuit class $\Class{AC}^0$, which is the class of problems
that can be defined by a family of unbounded fan-in circuits $(C_n)_{n\in
\mathbb{N}}$ with constant depth and size $n^{O(1)}$ and that may use
and-, or-, and not-gates, or the class $\Class{TC}^0$, which is defined
analogously, but where circuits can additionally use threshold-gates. Questions of uniformity will not be further
addressed in this paper, but we remark that \textsc{dlogtime} uniformity can be
achieved in all cases. Moreover, recall that $\Class{TC}^0 \subseteq \Class L
\subseteq \Class P$ holds. 

%== basic parameterized complexity, central problem
We use standard definitions from parameterized complexity, see for
instance~\cite{Cygan15, DowneyF99, FlumG06}. A \emph{parameterized problem} is a
set $Q\subseteq \Sigma^* \times \mathbb{N}$ for an alphabet~$\Sigma$. Given an
\emph{instance} $(x,k) \in \Sigma^* \times \mathbb{N}$, we call $x$ the
\emph{input} and $k$ the \emph{parameter}. We denote the parameterized version of a specific problem
by adding the prefix $\PLang[]{}$, so, for example, $\EM{dir}{\mathrm{del}}(\phi)$
becomes $\PEM{dir}{\mathrm{del}}(\phi)$.
%
%== Parameterized circuit classes
Similarly to the classical setting, we will establish the tractability of a parameterized problem by proving its containment in a certain parameterized circuit class, known as $\Para\Class{AC}^{0\uparrow}$:
 Formally, 
$\Para\Class{AC}^{0\uparrow}$ is defined as the class of all parameterized problems that can be
decided by a family of unbounded fan-in circuits $(C_{n, k})_{n, k\in
\mathbb{N}}$ that may use and-, or-, and not-gates, and have depth $f(k)$ and
size $f(k) \cdot n^{O(1)}$ for some computable function $f$. The class
$\Para\Class{AC}^{0}$ is defined analogously with the additional restriction that the depth of the circuit must be in $O(1)$. Note that $\Para\Class{AC}^{0}$-circuits can simulate threshold-gates when the
threshold is the parameter~\cite[Lemma 3.3]{BannachST15}. It
holds that $\Para\Class{AC}^{0} \subseteq \Para\Class{AC}^{0\uparrow} \subseteq
\Para\Class{P} = \Class{FPT}$, for a suitable notion of uniformity,
see~\cite{BannachST15}.
%== Reductions
We use $\Class{AC}^{0}$-many-one reductions in the
classical setting, and $\Para\Class{AC}^{0}$-many-one reductions in the
parameterized setting. Both are very weak notions of reducibility, which has
the benefit of yielding strong hardness results.

%%%%%%%%%%%%%%%%%%%%%%%%%%
% Parameterized
%%%%%%%%%%%%%%%%%%%%%%%%%%
\section{Parameterized Complexity of Relation Modification Problems}
\label{section:parameterized}
% \tcsautomoveaddto{main}{\subsection{Proofs for
% Section~\ref{section:parameterized}}}

In this section, we will address the parameterized complexity of relation
modification problems based on the first-order quantifier patterns of the
targeted properties.

\subsection{Undirected Graphs, Directed Graphs, and Arbitrary Structures}

We first establish the tractability landscapes of relation
modification problems with respect to undirected and directed graphs and
arbitrary structures. More specifically, we will show that these landscapes
coincide and exhibit the following complexities.

%{Pattern-Based Complexity Dichotomy for $\PEM{undir}{\otimes}(p)$ and $\PEM{dir}{\otimes}(p)$}
\begin{theorem}\label{theorem:p-undir}
  Let $p \in \{a,e\}^*$ be a pattern. For each modification operation $\otimes\in \{\mathrm{del},
  \mathrm{add}, \mathrm{edit}\}$, we have
  \begin{enumerate}
  \item $\PEM{arb}{\otimes}(p) \subseteq \Para\Class{AC}^{0\uparrow}$, if $p
    \preceq e^*a^*$.
  \item $\PEM{undir}{\otimes}(p)$ and $\PEM{dir}{\otimes}(p)$ both contain a $\Class{W}[2]$-hard
    problem, if $ae\preceq p$.
  \end{enumerate}
\end{theorem}

 We prove this theorem by individually establishing the respective upper and lower bounds. Note that the upper bound on arbitrary logical structures naturally carries over to the more restricted structures of basic, undirected, and directed graphs. Conversely, the hardness results for undirected and directed graphs imply hardness of more general structures.

\subsubsection*{Upper Bound.} 
We establish the tractability of the fragment $e^*a^*$ for all parameterized problem versions by proving the following lemma.

\begin{lemma}\label{lemma:p-undir-e*a*}
  For each modification operation $\otimes\in \{\mathrm{del}, \mathrm{add}, \mathrm{edit}\}$, we have $\PEM{arb}{\otimes}(e^*a^*) \subseteq
  \Para\Class{AC}^{0\uparrow}$.
\end{lemma}

\begin{proof}
  We generalize a construction given by Fomin et
  al.~\cite{FominGT20} such that it (a) handles arbitrary logical structures,
  and (b) is parallelized, showing the containment of the problems in $\Para\Class{AC}^{0\uparrow}$.

  For $\PEM{arb}{\mathrm{edit}}(\phi) \in \PEM{arb}{\mathrm{edit}}(e^*a^*)$, let
  $\phi$ be a first-order $\tau$-formula over a vocabulary $\tau = \{R_1, \dots, R_r\}$, where
  relation $R_i$ has arity $a_i$ and let furthermore $\phi$ be of the form $\exists x_1 \cdots \exists x_c
  \forall y_1 \cdots \forall y_d (\psi)$, where $\psi$ is a quantifier-free
  formula and $c,d \in \mathbb{N}$. We aim for an algorithm
  that, given a logical structure $\mathcal{A} = (A, R_1^{\mathcal{A}}, \dots,
  R_r^{\mathcal{A}})$ and $k \in \mathbb N$, decides in $\Para\Class{AC}^{0\uparrow}$ whether there is a
  modulator $S$ with $\|S\| \le k$ such that $\mathcal{A} \triangleright S \models \phi$. 
We start by observing the following: 

Let $v_1,\dots,v_c$ be a sequence of elements from the universe~$A$. Then, either $\mathcal{A} \models (\forall y_1 \cdots \forall y_d(\psi))[v_1, \dots, v_c]$ (in which case we accept the instance) or
$   \mathcal{A}\ \not\models\ (\forall y_1 \cdots \forall y_d(\psi))[v_1, \dots, v_c]$.
  In the latter case, there must be a sequence $u_1, \dots, u_d$ of elements from the universe~$A$ such that the structure induced by the union of $v_1, \dots, v_c$ and $u_1, \dots, u_d$ violates the
  assertions made by~$\psi$. If we can edit at most
  $k$ tuples to resolve all violations on the structure induced by $v_1, \dots, v_c$ and every violating assignment to the universal variables, then $\mathcal{A} \triangleright S \models \phi$. Otherwise, we can rule out $v_1,\dots,v_c$ as an assignment to $x_1, \dots, x_c$.

  Based on this observation, we construct the following algorithm: Consider all of the $|A|^c$ possible assignments $(v_1, \dots, v_c)$ to
  the variables $x_1,\dots,x_c$ in parallel. For each, we in turn consider all of the $2^{\sum_{i =
  1}^{r} c^{a_i}}$ possible relations on the structure induced by
  $\{v_1,\dots,v_c\}$ in parallel. For each such set of relations, there is a modulator~$S^{(0)}$
  such that $\mathcal{A} \triangleright S^{(0)}$ yields precisely this set of relations.
  By definition, $S^{(0)}$ is a tuple in which each entry is a set of tuples over (subsets of) $\{v_1,
  \dots,v_c\}$, hence, its
  size is constantly bounded. Now, fix a certain set of relations on $\{v_1, \dots, v_c\}$ and a respective modulator~$S^{(0)}$.
  We use $k$ layers of a Boolean circuit to find and
  resolve all violations that arise from assignments to the universally
  quantified variables. We let $S^{(0)}$ be associated with the first layer, and proceed as follows in all subsequent layers~$i \in \{1, \dots, k\}$: We check whether an assignment $(w_1, \dots, w_d)$ to the variables $y_1, \dots, y_d$
  (for example, the lexicographically first) exists such that
  $\mathcal{A} \triangleright S^{(i)}\ \not\models\ \psi[v_1, \dots, v_c, w_1, \dots,
  w_d]$. If we cannot find such an assignment, we can accept the instance, since we have
  $\mathcal{A} \triangleright S^{(i)}\models\phi$. Otherwise, we have to edit a
  relation in the structure induced by the current assignment. We branch over all possible edits and associate each branch with a new modulator $S^{(i + 1)}$ that extends $S^{(i)}$ by the relations we marked to be edited. If at layer $k$, we still find an assignment $(w_1, \dots, w_d)$ to the variables $y_1, \dots, y_d$ such that
  $\mathcal{A} \triangleright S^{(k)}\ \not\models\ \psi[v_1, \dots, v_c, w_1, \dots,
  w_d]$, then we reject the instance.

  Since $d$, $c$ and all the arities $a_i$ are constants, the number of branches created at each level of the search tree is constant as well. Therefore, the total size of the search
  tree is at most $2^{\sum_{i = 1}^{r} c^{a_i}} + g(r, c, d, a_1, \dots,
  a_r)^k$, where $g$ is the constant function that describes the number of branches. The depth of the
  search tree is bounded by the number of tuples we can edit, which is our
  parameter. Hence, we get a $\Para\Class{AC}^{0\uparrow}$-circuit.
\end{proof}

\subsubsection*{Lower Bounds.} To show hardness of the fragment $ae$ for various problem versions, we
perform a reduction from $\PLang{Set-Cover}$, a well-known $\Class{W}[2]$-hard
problem~\cite[page~464]{DowneyF99}, which we here formulate as a graph problem
on the incidence graph of the input set family.\footnote{This variant is also known as
$\PLang{Red-Blue Dominating Set}$ in Parameterized Complexity.} Recall that we write
$u\undiradj s$ to denote $\{(u, v), (v, u)\} \subseteq E$.

%\begin{problem}[{$\PLang{set-cover}$}]
  \begin{parameterizedproblem}
  \problemtitle{$\PLang{Set-Cover}$}
    \probleminput{An undirected bipartite graph $G = (S \mathbin{\dot\cup} U, E)$
      with shores $S$ and~$U$ and a number~$k$.}
    \problemparameter{$k$}
    \problemquestion{ Is there a cover $C \subseteq S$ with $|C| \le k$ of~$U$, meaning that for each $u \in U$ there is an $s \in  C$
    with $u\undiradj s$?} 
  \end{parameterizedproblem}
%\end{problem}

All our parameterized hardness results that are based on reductions from the set cover problem consist of six steps:
\begin{enumerate}
  \item We first establish the hardness for the problem versions that restrict the type of structure to undirected graphs. In order to do this, we start by constructing a formula~$\phi$ with the quantifier
    pattern $p$ for which we want to establish hardness. 
%    To keep the presentation concise, we only
%    state the formula for undirected graphs. Note that thanks to our definition
%    of cost and undirected graphs being special cases of directed graphs, we can, given a formula~$\phi$, easily construct a semantically equivalent formula for directed graphs by replacing each term $x\adj y$ by the term
%    $(x\adj y) \land (y\adj x)$.
  \item We describe a construction that transforms an arbitrary input instance $(G, k)$ of $\PLang{Set-Cover}$ into an instance $(G', k')$ of $\PEM{\mathrm{undir}}{\mathrm{del}}(\phi)$ (or $\PEM{\mathrm{undir}}{\mathrm{add}}(\phi)$, respectively), typically by adding new vertices and~edges.
  \item We show that $(G, k) \in \PLang[]{Set-Cover}$ implies $(G', k')
  \in \PEM{\mathrm{undir}}{\mathrm{del}}(\phi)$ (or $(G', k')
  \in \PEM{\mathrm{undir}}{\mathrm{add}}(\phi)$, resp.). We call this part of the proof the \emph{forward direction}.
  \item Next, we show that $(G', k')
  \in \PEM{\mathrm{undir}}{\mathrm{del}}(\phi)$ (or $(G', k')
  \in \PEM{\mathrm{undir}}{\mathrm{add}}(\phi)$, resp.) implies $(G, k) \in
  \PLang[]{Set-Cover}$. We call this part of the proof the \emph{backward direction}.
  \item We show that both directions of the reduction are still correct when the problem $\PEM{\mathrm{undir}}{\mathrm{edit}}(\phi)$ is considered instead.
  \item Finally, we show how each reduction can be adapted to
  provide a lower bound for $\PEM{dir}{\otimes}(\phi)$ with $\otimes\in \{\mathrm{del},
  \mathrm{add}, \mathrm{edit}\}$, and thus also for
  $\PEM{arb}{\otimes}(\phi)$.
\end{enumerate}

%This has the advantage that the reduction and forward direction is the same for
%both the deletion or addition problem and the editing problem, and only the
%backward direction has to be adapted. Furthermore,
Due to Observation~\ref{observation:del-add-same}, that states that the complexities of
edge deletion and edge addition coincide, it suffices to analyze one of these problem versions in our proofs. We will usually choose the one that eases the presentation and state
our choice at the beginning of each \emph{reduction} section.

\begin{lemma}\label{lemma:p-undir-ae} For each modification operation $\otimes\in \{\mathrm{del},
  \mathrm{add}, \mathrm{edit}\}$, $\PEM{undir}{\otimes}(ae)$ contains a
  $\Class{W}[2]$-hard problem.
\end{lemma}

\begin{proof}
  \emph{The formula.} Consider the formula 
  \begin{align*}
    \phi_{\ref{lemma:p-undir-ae}} = \forall x \exists y
    \bigl(\neg x\adj x \to (x\adj y \land \neg y\adj y)\bigr), \label{lemma:p-undir-ae-eq}
  \end{align*}
  which exhibits the pattern $ae$ and expresses that every vertex without a self-loop has a neighbor without a self-loop.

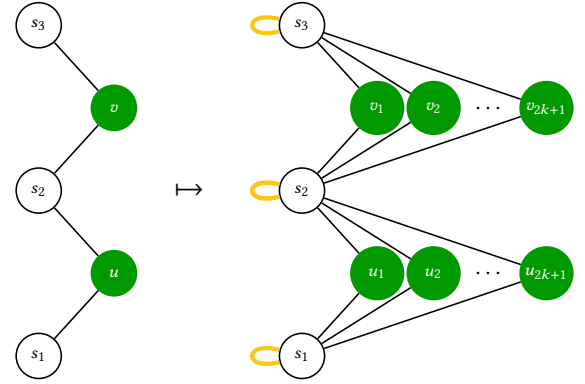
\begin{figure}
  \centering
  \begin{tikzpicture}[
      right part/.style={xshift=3.5cm},y=2.2cm,
      node/.append style={minimum size=6mm, inner sep=0pt},
      large/.style={minimum size=7mm},
    ]
    
    % Sets
    \foreach \i in {1,2,3} {
      \node (s\i) at (0,\i) [node] {$s_\i$};

      \node (rs\i) at (0,\i) [right part, node] {$s_\i$};
    }

    % Universe
    \foreach \i/\name in {1/u,2/v} {
      \node (\name)   at (0,\i) [green node,shift={(1,.5)}] {$\name$};

      \node (\name1)  at (0,\i) [green node, large]  [right part,shift={( 1,.5)}] {$\name_1$};
      \node (\name2)  at (0,\i) [green node, large]  [right part,shift={( 1.75,.5)}] {$\name_2$};
      \node (\name3)  at (0,\i) []  [right part,shift={( 2.5,.5)}] {$\dots$};
      \node (\name4)  at (0,\i) [green node, large]  [right part,shift={( 3.25,.5)}] {$\name_{2k+1}$};
    }

    \foreach \s/\u/\anc in {1/u/south,2/u/north,2/v/south,3/v/north} {
      \draw (s\s) -- (\u);
      \draw [delete edge] (rs\s) to[loop left] (rs\s);

      \draw  (rs\s) -- (\u1);
      \draw  (rs\s) -- (\u2);
      \draw  (rs\s) -- (\u4);
    }

    % arrow
    \node (arrow) at (2,2) {\Large$\mapsto$};
  \end{tikzpicture}
  \caption{Visualization of the reduction in
    Lemma~\ref{lemma:p-undir-ae}. }
  \label{figure:p-undir-ae}
  \Description{Picture of an exemplary reduction, which maps a bipartite graph $G =
  (\{s_1,s_2,s_3\}\mathbin{\dot\cup}\{u,v\}, E)$ to a graph $G'=(V',E')$ (the indicated   
  colors, numbers, and labels are not part of the input or output). To satisfy
  the formula $\phi_{\ref{lemma:p-undir-ae}}$, some of the yellow edges need to
  be deleted. For each element of $U = \{u,v\}$, exactly $2k+1$ vertices are
  added to the new graph. Each undirected edge $u \undiradj s$ ($v \undiradj s$,
  respectively) gets replaced by $2k + 1$ undirected edges, namely $u_i
  \undiradj' s$ ($v_i \undiradj' s$, resp.) for all copies $u_i$ of~$u$ (all
  copies $v_i$ of~$v$, resp.). The size-$1$ set cover $C = \{s_2\}$ corresponds
  to the fact that deleting the edge $s_2 \undiradj' s_2$ from $G'$ yields a
  graph in which each vertex without a self-loop has a neighbor without a
  self-loop. The same is true for the size-$2$ set cover $C = \{s_1, s_3\}$. In
  contrast, $C = \{s_1\}$ is not a set cover as $v$ is not covered and, indeed,
  all neighbors of~$v_1$ (namely $s_2$ and $s_3$) have a self-loop, if we delete
  neither $s_2\undiradj' s_2$ nor $s_3\undiradj' s_3$. Note that our
  construction is slightly more complex than the deletion case demands, but we
  will benefit from it when we lift our results to the editing case later on.}
\end{figure}

  \emph{The reduction.} We  reduce $\PLang{Set-Cover}$ to the edge deletion problem,
  that is, to $\PEM{undir}{\mathrm{del}}(\phi_{\ref{lemma:p-undir-ae}})$. Let $(S
  \mathbin{\dot\cup} U, E, k)$ be given as
  input. The reduction outputs $k' = k$ together with the undirected
  graph $G' = (V', E')$ constructed as follows:
  \begin{itemize}
    \item For each $s \in S$, add $s$ to $V'$.
    Call this set of vertices $S'$.
    \item For each $u \in U$, add $2k+1$ copies $u_1,\dots, u_{2k+1}$ of~$u$
      to~$V'$. Call this set of vertices $U'$.
    \item Whenever $u \undiradj s$ holds for $u \in U, s \in S$, let all $u_i$ be adjacent to $s$
      in the new graph, that is, for $i\in\{1,\dots, 2k+1\}$, let $u_i \undiradj'
      s$. 
    \item For every $s \in S'$ in the new graph, add the self-loop $s \undiradj' s$.
  \end{itemize}
  An exemplary reduction is shown in Figure~\ref{figure:p-undir-ae}, which maps a bipartite graph $G =
  (\{s_1,s_2,s_3\}\mathbin{\dot\cup}\{u,v\}, E)$ to a graph $G'=(V',E')$ (the indicated   
  colors, numbers, and labels are not part of the input or output). To satisfy
  the formula $\phi_{\ref{lemma:p-undir-ae}}$, some of the yellow edges need to
  be deleted. For each element of $U = \{u,v\}$, exactly $2k+1$ vertices are
  added to the new graph. Each undirected edge $u \undiradj s$ ($v \undiradj s$,
  respectively) gets replaced by $2k + 1$ undirected edges, namely $u_i
  \undiradj' s$ ($v_i \undiradj' s$, resp.) for all copies $u_i$ of~$u$ (all
  copies $v_i$ of~$v$, resp.). The size-$1$ set cover $C = \{s_2\}$ corresponds
  to the fact that deleting the edge $s_2 \undiradj' s_2$ from $G'$ yields a
  graph in which each vertex without a self-loop has a neighbor without a
  self-loop. The same is true for the size-$2$ set cover $C = \{s_1, s_3\}$. In
  contrast, $C = \{s_1\}$ is not a set cover as $v$ is not covered and, indeed,
  all neighbors of~$v_1$ (namely $s_2$ and $s_3$) have a self-loop, if we delete
  neither $s_2\undiradj' s_2$ nor $s_3\undiradj' s_3$. Note that our
  construction is slightly more complex than the deletion case demands, but we
  will benefit from it when we lift our results to the editing case later on.
  
  \emph{Forward direction.}
  Let  $(S \mathbin{\dot\cup} U, E, k) \in \PLang{Set-Cover}$ be given. We need
  to show that $(G', k') \in
  \PEM{undir}{\mathrm{del}}(\phi_{\ref{lemma:p-undir-ae}})$ holds. Let $C
  \subseteq S$ with $|C| \le k$ be a cover of~$U$. Without loss of generality, assume
  that every vertex in $C$ has at least one neighbor in $U$. Let $C' = \{(s, s) \mid s\in C\}$. We claim that $(V', E'\setminus C') \models
  \phi_{\ref{lemma:p-undir-ae}}$, that is, by removing all self-loops from
  vertices $s \in C$, every vertex without a self-loop has a neighbor without a
  self-loop and thus the resulting graph satisfies~$\phi_{\ref{lemma:p-undir-ae}}$. To see this, call a vertex that has a
  self-loop or a neighbor without a self-loop \emph{happy}. Naturally, if every
  vertex of the graph $(V', E'\setminus C')$ is happy, we have $(V', E'\setminus C')
  \models \phi_{\ref{lemma:p-undir-ae}}$. Observe that all vertices in $S'$ are happy,
  since they have at least one neighbor in $U'$ and the vertices in $U'$ do not have
  self-loops. Moreover, the vertices in $U'$ each have at least one neighbor in $S'$
  that does not have a self-loop since we assumed $C$ to be a cover of~$U$. Clearly, $\|C'\| = |C| \leq k$.

  \emph{Backward direction.} Conversely, suppose that for $G' = (V', E')$, we are
  given a set $D\subseteq E'$ with $\|D\| \leq k$ such that $(V', E'\setminus D)
  \models \phi_{\ref{lemma:p-undir-ae}}$. Observe that each $u\in U$ is only
  connected to vertices in $S$ that have a self-loop. Hence, from $(V', E'\setminus D)\models \phi_{\ref{lemma:p-undir-ae}}$, we conclude that for each $u_i\in U'$, there is an $s\in S'$ with
  $s\undiradj' u_i$ such that $(s, s) \in D$. This means that $C = \{s\in S' \mid
  (s, s) \in D\}$ is a set cover of $(S \mathbin{\dot\cup} U, E)$ and $|C| \le \|D\| \leq k$.

  \emph{Edge editing.} To see that the reduction works for edge editing as well,
  observe that there are only two other modifications that will make vertices in~$U'$ happy:
  ($1$) We can add a self-loop to such vertices, or ($2$) we can make vertices in~$U'$ adjacent to each other, where edges may be added both between vertices that stem from the same vertex in~$U$, and between vertices that are copies of two distinct vertices in~$U$.
  However, since introducing a self-loop can only make
  one more vertex happy, and introducing an edge between two distinct vertices in $U'$ can only make up to
  two vertices happy, these operations do not suffice to ensure the happiness of all $2k + 1$ copies of $u$.
Consequently, we
  still have to delete the self-loop of at least one neighbor of $u$.
  
  \emph{Adaption to the directed setting.} We choose
  \begin{align*}
    \phi_{\ref{lemma:p-undir-ae}, \text{dir}} = \forall x \exists y
    \bigl(\neg x\adj x \to (x\adj y \land \neg y\adj y)\bigr),
  \end{align*}
  and perform the reduction in the same way, but let each edge~$u_i -' s$ that we added be directed from $u_i$ to $s$.
\end{proof}

\subsection{Basic Graphs}

In the last section, we proved that the complexity landscapes of relation modification problems coincide on undirected graphs, directed graphs, and arbitrary structures.
We now consider the restriction of input and target structures to basic graphs,
that is, undirected graphs that have no self-loops, which, interestingly, yields a different tractability frontier than the previous problem versions.

% {Pattern-Based Complexity Dichotomy for $\PEM{basic}{\otimes}(p)$}
\begin{theorem}\label{theorem:p-basic}
  Let $p \in \{a,e\}^*$ be a pattern. For each modification operation $\otimes\in \{\mathrm{del},
  \mathrm{add}, \mathrm{edit}\}$, we have
  \begin{enumerate}
  \item $\PEM{basic}{\otimes}(p) \subseteq \Para\Class{AC}^{0\uparrow}$, if $p
    \preceq e^*a^*$ or $p \preceq ae$.
  \item $\PEM{basic}{\otimes}(p)$ contains a $\Class{W}[2]$-hard
    problem, if $aea \preceq p$, $eae \preceq p$, $aee \preceq p$, or $aae \preceq p$.
  \end{enumerate}
\end{theorem}

\subsubsection*{Upper Bound.} Contrary to the more general structures analyzed in the previous section, the
fragment $ae$ becomes tractable on basic graphs.

\begin{lemma}\label{lemma:p-basic-ae}
  For each modification operation $\otimes\in \{\mathrm{del}, \mathrm{add}, \mathrm{edit}\}$, we have
  $\PEM{basic}{\otimes}(ae) \subseteq \Para\Class{AC}^{0} \subseteq \Para\Class{AC}^{0\uparrow}$.
\end{lemma}

\begin{proof} We present the proof for edge deletion and then argue that
  the statement holds for edge editing as well.
  Let $\phi$ be of the form $\forall x\exists y (\psi)$, where $\psi$ is
  quantifier-free.  Assume without loss of
  generality that $\psi$ is in disjunctive normal form, that is,
  \begin{align*}
    \phi=\forall x \exists y \bigvee (\ell_1 \land \ell_2 \land \cdots \land \ell_m),
  \end{align*}
  where each $\ell_i$ is a literal of the form $x \adj y$, $x \nadj y$, $x = y$, or $x \neq y$ (note that on basic graphs, the terms $x = x$, $y = y$, $x\neq x$, $y\neq y$, $x\adj
  x$, $y\adj y$, $x\nadj x$, and $y\nadj y$ can be replaced by true or false, and $\adj$ and $=$ are both symmetric). Also observe that when a disjunct contains both a literal $\ell$ and
  its negation, it cannot be satisfied, so it may also be replaced by
  false (and, hence, removed).
  Since the only available relations are $=$ and~$\adj$, it remains to 
  consider the unit disjuncts as well as the four different disjuncts $x = y
  \land x\adj y$, $x = y \land x\nadj y$, $x \neq y \land x\adj y$, and $x \neq
  y \land x\nadj y$. Out of these, the first one cannot be
  satisfied in basic graphs, the second one is
  equivalent to $x=y$, and the third one equivalent to $x \adj  y$ on basic graphs.
 Thus, besides the four unit
  disjuncts, only $x \neq
  y \land x\nadj y$ remains. We may assume that a given graph has at least
  two vertices. 

  Suppose $\psi$ contains the disjunct $x = y$ or the disjunct
  $x\nadj y$. Since the quantifier prefix is $\forall x\exists y$,
  interpreting $y$ identically to~$x$ will make the disjunct true and,
  thus, also~$\phi$. Next, if $x \neq y$ is present as a
  disjunct, we can also satisfy the disjunct and thus $\phi$ by interpreting $x$ and $y$ differently. 

  It remains to consider $\psi$ consisting of just $x \adj y$ or of just $x
  \neq y \land x\nadj y$ or of their disjunction. If $x \adj y$ is the
  only disjunct, the formula is 
  true if there is no isolated vertex in the graph; and if there is an isolated
  vertex, deleting edges cannot change this fact. 
%  Now, we consider disjuncts consisting of two literals. If the literal $x = y$
%  occurs in the disjunct, we can again precompute the truth value of the
%  disjunct: If the other literal is $x\adj y$, the disjunct can be substituted
%  by false, and if the other literal is $x\nadj y$, it can be substituted by
%  true. The disjunct $x \neq y \land x\adj y$ can be replaced by the same truth
%  value as $x\adj y$ in the setting of basic graphs.
  If $x \neq y \land x\nadj y$ is the only disjunct, the formula
  specifies that there are no universal vertices, that is, vertices that are adjacent to all others. If this is not yet the
  case, we can ``turn'' up to $2k$ universal
  vertices into non-universal vertices by deleting $k$ edges; thus, we need to check
  whether there are at most $2k$ universal vertices, which can be done with an $\Para\Class{AC}^{0}$-circuit. The final case is $\phi
  = \forall x \exists y((x \adj y) \lor (x \neq y \land x\nadj y))$,
  but this formula is trivially true (on graphs with at least two
  vertices). 
  
%%   Lastly, we turn to the disjunction of disjuncts. The only disjuncts that
%%   cannot be replaced by true or false directly are $x\adj y$ (or $(x \neq y
%%   \land x\adj y)$, respectively) and $(x \neq y \land x\nadj y)$. This means
%%   that the only formula that cannot be seen to be precomputed is $(x \neq y
%%   \land x\adj y) \lor (x \neq y \land x\nadj y)$. But this formula is again
%%   trivially true if the graph consists of at least two vertices.

  We can extend the above arguments to the task of editing edges quite
  easily: All cases stay the
  same, with the exception that violations of $\forall x \exists y
  (x\adj y)$ can now also be repaired by adding edges if there are only up to $2k$ isolated vertices. Once more,
  this can be checked using a $\Para\Class{AC}^{0}$-circuit. 
\end{proof}

\subsubsection*{Lower Bounds.} We now go on to prove lower bounds for basic graphs.
Observe that since we have shown that $\PEM{basic}{\otimes}(p) \subseteq \Para\Class{AC}^{0\uparrow}$ for $p
    \preceq e^*a^*$ and $p \preceq ae$, the next larger quantifier patterns to be addressed are $p \in \{aea, aee, aae, eae\}$. In the following, we show that for all these choices of $p$, $\PEM{basic}{\otimes}(p)$ contains a $\Class{W}[2]$-hard problem by performing reductions from $\PLang{Set-Cover}$.

\begin{lemma}[$\star$]\label{lemma:p-basic-aea}
  For each $\otimes\in \{\mathrm{del}, \mathrm{add}, \mathrm{edit}\}$, $\PEM{basic}{\otimes}(aea)$ contains a $\Class{W}[2]$-hard problem.
\end{lemma}

\begin{lemma}[$\star$]\label{lemma:p-basic-aee}
  For each $\otimes\in \{\mathrm{del}, \mathrm{add}, \mathrm{edit}\}$, $\PEM{basic}{\otimes}(aee)$ contains a $\Class{W}[2]$-hard problem.
\end{lemma}

\begin{lemma}[$\star$]\label{lemma:p-basic-aae}
  For each $\otimes\in \{\mathrm{del}, \mathrm{add}, \mathrm{edit}\}$, $\PEM{basic}{\otimes}(aae)$ contains a $\Class{W}[2]$-hard problem.
\end{lemma}

\begin{lemma}\label{lemma:p-basic-eae}
  For each $\otimes\in \{\mathrm{del}, \mathrm{add}, \mathrm{edit}\}$, $\PEM{basic}{\otimes}(eae)$ contains a $\Class{W}[2]$-hard problem.
\end{lemma}

\begin{proof}
  \emph{The formula.} We consider the formula
  \begin{align*}
    \phi_{\text{radius-2}} = \exists c \forall x \exists y
    \bigl(x = c \lor x\adj c \lor (x\adj y \land y\adj c)\bigr), %\label{lemma:p-basic-eae-eq} 
  \end{align*}
  which expresses that a graph's radius is at most~$2$.
  
  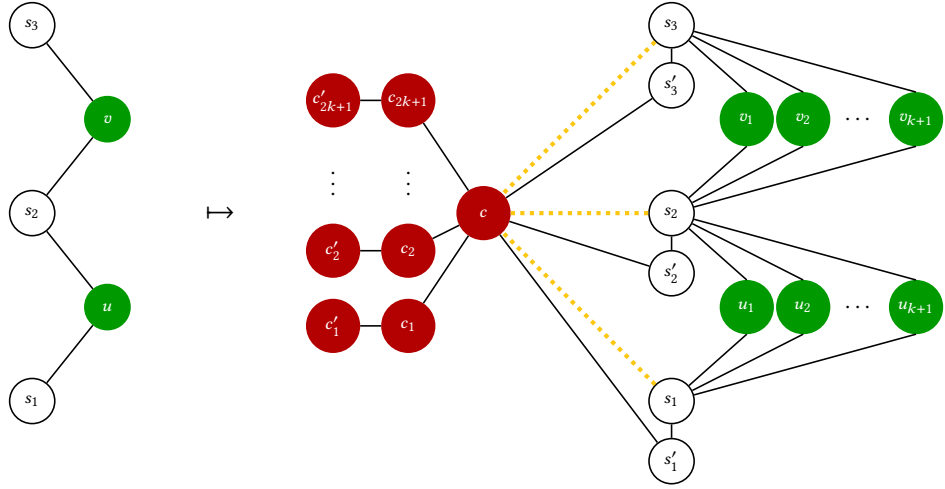
\begin{figure*}
    \centering
    \begin{tikzpicture}[
        right part/.style={xshift=8.5cm},y=2.5cm,
        node/.append style={minimum size=6mm, inner sep=0pt},
        large/.style={minimum size=7mm},
      ]

      \node (c) at (0,2) [red node, right part,shift={( -25mm,0)}, large] {$c$};
      \node (c11) at (-1,1.4) [red node, right part,shift={( -25mm,0)}, large] {$c_1$};
      \node (c12) at (-1,1.8) [red node, right part,shift={( -25mm,0)}, large] {$c_2$};
      \node (c13) at (-1,2.2) [right part,shift={( -25mm,0)}] {$\vdots$};
      \node (c14) at (-1,2.6) [red node, right part,shift={( -25mm,0)}, large] {$c_{2k +
      1}$};
      \node (c21) at (-2,1.4) [red node, right part,shift={( -25mm,0)}, large] {$c_1'$};
      \node (c22) at (-2,1.8) [red node, right part,shift={( -25mm,0)}, large] {$c_2'$};
      \node (c23) at (-2,2.2) [right part,shift={( -25mm,0)}] {$\vdots$};
      \node (c24) at (-2,2.6) [red node, right part,shift={( -25mm,0)}, large] {$c_{2k +
      1}'$};

      \draw 
      (c21) -- (c11)
      (c22) -- (c12)
      (c24) -- (c14)
      (c11) -- (c)
      (c12) -- (c)
      (c14) -- (c);

      \foreach \i in {1,2,3} {
        \node (s\i) at (0,\i) [node] {$s_\i$};
  
        \node (rs\i')    at (0,\i) [node]      [right part,shift={( 0mm,-8mm)}] {$s_\i'$};
        \node (rs\i)    at (0,\i) [node]      [right part] {$s_\i$};
  
        \draw [add edge] (rs\i) -- (c);
        \draw  (rs\i') -- (c);
        \draw  (rs\i) -- (rs\i');
      }

      \foreach \i/\name in {1/u,2/v} {
        \node (\name)   at (0,\i) [green node,shift={(1,.5)}] {$\name$};
  
        \node (\name1)  at (0,\i) [green node, large]  [right part,shift={( 1,.5)}] {$\name_1$};
        \node (\name2)  at (0,\i) [green node, large]  [right part,shift={( 1.75,.5)}] {$\name_2$};
        \node (\name3)  at (0,\i) []  [right part,shift={( 2.5,.5)}] {$\dots$};
        \node (\name4)  at (0,\i) [green node, large]  [right part,shift={( 3.25,.5)}] {$\name_{k+1}$};
      }
  
      \foreach \s/\u/\anc in {1/u/south,2/u/north,2/v/south,3/v/north} {
        \draw (s\s) -- (\u);
  
        \draw 
        (rs\s) -- (\u1.\anc)
        (rs\s) -- (\u2.\anc)
        (rs\s) -- (\u4.\anc);
      }
  
      % arrow
      \node (arrow) at (2.5,2) {\Large$\mapsto$};
    \end{tikzpicture}
    \caption{Visualization of the reduction in Lemma~\ref{lemma:p-basic-eae}.}
    \label{figure:p-basic-eae}
    \Description{An exemplary reduction. We are using
    the same conventions as in the previous figure. For each $s\in S$, we add
    vertices $s, s'$ with $s\undiradj' s'$, and for each element of $U$ we add exactly
    $k+1$ copies. Each edge $s \undiradj u$ gets replaced
    by edges connecting $s$ to every copy $u_i$ of~$u$. Moreover, a vertex $c$ is
    added and connected to every $s'$, vertices $c_1, \dots, c_{2k+1}$ are added
    and connected to $c$, and $c_1', \dots, c_{2k+1}'$ are added such that $c_i \undiradj'
    c_i'$ for all $i\in \{1, \dots, 2k+1\}$. The size-$1$ set cover $C =
    \{s_2\}$ corresponds to the fact that adding the edge $c \undiradj' s_2$ yields a
    graph in which every vertex has distance at most $2$ from $c$. The same is true
    for the size-$2$ set cover $C = \{s_1,s_3\}$. In contrast, $C = \{s_1\}$ is
    not a set cover, as $v$ is not covered and, indeed, $v_1$ has distance $3$
    from $c$, if we add neither $s_2\undiradj' c$ nor $s_3\undiradj' c$. The argument that $c$
    is the only viable center of the graph is given in the proof.}
  \end{figure*}

  \emph{The reduction.} We will reduce to the edge addition problem. Let $(S
  \mathbin{\dot\cup} U, E,k)$ be given as input.  The reduction outputs $k' = k$
  and a basic graph $G' = (V',E')$ constructed as follows:
  \begin{itemize}
    \item For each $s_i \in S$, add vertices named $s_{i}, s_{i}'$ to $V'$ and
    connect them via an edge.
    \item Add a vertex $c$ to $V'$ and connect it to every $s_{i}'$. Add $2k + 1$ vertices $c_i$ to $V'$ and
      connect them to $c$, and add $2k + 1$ vertices $c_i'$ to $V'$ and connect
      $c_i'$ to $c_i$.
    \item For each $u \in U$, add $k+1$ copies $u_1,\dots,
      u_{k+1}$ to $V'$.
    \item Whenever $u \undiradj s_i$ holds, connect all $u_i$ to $s_i$ in the new
      graph, that is, for $i\in\{1,\dots, k+1\}$ let $u_i \undiradj' s_i$. 
  \end{itemize}
  An example for the reduction is shown in Figure~\ref{figure:p-basic-eae}, using
    the same conventions as in the previous figure. For each $s\in S$, we add
    vertices $s, s'$ with $s\undiradj' s'$, and for each element of $U$ we add exactly
    $k+1$ copies. Each edge $s \undiradj u$ gets replaced
    by edges connecting $s$ to every copy $u_i$ of~$u$. Moreover, a vertex $c$ is
    added and connected to every $s'$, vertices $c_1, \dots, c_{2k+1}$ are added
    and connected to $c$, and $c_1', \dots, c_{2k+1}'$ are added such that $c_i \undiradj'
    c_i'$ for all $i\in \{1, \dots, 2k+1\}$. The size-$1$ set cover $C =
    \{s_2\}$ corresponds to the fact that adding the edge $c \undiradj' s_2$ yields a
    graph in which every vertex has distance at most $2$ from $c$. The same is true
    for the size-$2$ set cover $C = \{s_1,s_3\}$. In contrast, $C = \{s_1\}$ is
    not a set cover, as $v$ is not covered and, indeed, $v_1$ has distance $3$
    from $c$, if we add neither $s_2\undiradj' c$ nor $s_3\undiradj' c$. The argument that $c$
    is the only viable center of the graph is given in the proof.

  \emph{Forward direction.}
  Let  $(S \mathbin{\dot\cup} U, E, k) \in \PLang{Set-Cover}$ be given. We show that $(G', k') \in
  \PEM{basic}{\mathrm{add}}(\phi_{\text{radius-2}})$ holds. Let $C
  \subseteq S$ with $|C| \le k$ be a cover of~$U$ and let $C' = \{(s, c) \mid s\in C\}\cup
  \{(c, s) \mid s\in C\}$. We claim that $(V', E' \cup C') \models
  \phi_{\text{radius-2}}$, that is, adding all $s \undiradj' c$ to~$E'$
  yields a graph of radius~$2$. Let us check all vertices in~$V'$: Choose $c$ to be the first existentially-quantified vertex
  (the center). Then, each
  $s'$ is adjacent to $c$, and each $s$ has distance at most $2$ to
  $c$ via $s'$. Furthermore, each $c_i$ is also adjacent to $c$ and each
  $c_i'$ has distance at most $2$ to $c$ via $c_i$. It remains to check $u \in U$. Since $C$ is a cover, there is an $s\in C$ for every $u\in
  U$ with $s\undiradj' c$, and by construction, every vertex in $U$ has distance at most $2$ to $c$
  via $u_i \undiradj' s\undiradj' c$.

  \emph{Backward direction.}
  Conversely, suppose that for $G' = (V', E')$ we are given a set $A \subseteq
  V' \times V'$ with $\|A\| \le k$ such that $(V', E'\cup A) \models
  \phi_{\text{radius-2}}$. 
  First, assume that $c$ is assigned to the first existentially-quantified
  variable. Then, for each $u \in U$, consider the copies $u_i$
  for $i\in \{1,\dots,k+1\}$. Since we are only allowed to add $k$
  edges, we cannot afford to directly connect each $u_i$ to~$c$,$s$ or~$s'$. We are thus forced to add $s \undiradj'c$
  for some $s$ with $u_i \undiradj' s$. Therefore, each $u\in U$
  has a neighbor $s\in S$ such that $s\undiradj' c$, hence, the set $C = \{s \mid s\in
  S \text{ and }
  (s, c) \in A\}$ is a cover with $|C| \leq \|A\| \le k$.

  Now, we argue that we cannot assign the first quantified variable differently. Suppose we chose $x$ to
  be a vertex from~$c_i$. Then each $c_j'$ with $i\neq j$ has distance
  greater than $2$ from $c_i$, and since the $c_j'$ have no common neighbors, we would have
  to add more than $k$ edges to bound their distances to be at most
  $2$ from $c_i$. The argument when $x$ is chosen to be one of the $s$, $s'$, or
  $u$ is similar.
 
  \emph{Edge editing.}
  To see that the reduction also works for edge editing, observe that
  deleting edges cannot reduce the radius of a graph.

  \emph{Adaption to the directed setting.} 
  The reduction can be performed in exactly the same way. Choosing a set $S_i$ now
  corresponds to adding the directed edge from $s_i$ to $c$.
\end{proof}

As a corollary of the proof of Lemma~\ref{lemma:p-basic-eae}, we obtain two hardness results that were previously unknown.

\begin{corollary}
  The problem $\PLang{Edge-Adding To Radius 2}$ as well as $\PLang{Edge-Editing To Radius 2}$ is $\Class{W}[2]$-hard.
\end{corollary}

These statements can be extended to arbitrary fixed radii.

\begin{lemma}\hfil
  \begin{enumerate}
    \item $\PLang{Edge-Adding To Radius 1}$ as well as $\PLang{Edge-Editing}$ $\Lang{To Radius
    1}$ lie in $\Para\Class{AC}^0$.
    \item For every fixed $r \geq 2$, both
    $\PLang{Edge-Adding To Radius $r$}$ and
    $\PLang{Edge-Editing To Radius $r$}$ are $\Class{W}[2]$-hard.
  \end{enumerate}  
\end{lemma}

\begin{proof}
  To show (1), notice that we can test for all
  potential centers $c$ in parallel and count for every $c$ how many vertices are not
  adjacent to it. This yields a $\Para\Class{AC}^0$ algorithm.

  To show (2), extend the construction from Lemma~\ref{lemma:p-basic-eae} by replacing the
  paths $c\undiradj' c_i\undiradj' c_i'$ by paths $c\undiradj' c_i^1\undiradj' c_i^2
  \undiradj' \cdots \undiradj' c_i^{r-1}$ and instead of connecting $u_i$ directly
  to $s$ if $u\undiradj s$, add the path $u_i \undiradj' u_i^1\undiradj' u_i^2\undiradj' \cdots
  \undiradj' u_i^{r - 2} \undiradj' s$.
\end{proof}

\subsection{Monadic Structures}

To complete the classification of the parameterized complexity, we finally turn to structures over monadic
vocabularies.

\begin{theorem}\label{theorem:p-monadic}
  Let $p \in \{a,e\}^*$ be a pattern. For each modification operation $\otimes\in \{\mathrm{del},
  \mathrm{add}, \mathrm{edit}\}$, we have $\PEM{mon}{\otimes}(p) \subseteq \Para\Class{AC}^0$.
\end{theorem}

\begin{proof}
  For a fixed monadic vocabulary $\tau = \{R_1^1, \dots, R_r^1\}$ and a given
  structure $\mathcal{A} = (A, R_1^{\mathcal{A}}, \dots, R_r^{\mathcal{A}})$,
  let the \emph{type} of an element $a\in A$ be defined as
  $\operatorname{type}(a) \coloneq \{R_i\in \tau \mid a\in R_i\}$. Let $T_\tau$ denote the set
  of all possible types of elements in $\mathcal{A}$ and observe that $|T_\tau| = 2^r$. We can order $T_\tau$ by letting $t_1
  <_{\text{lex}} t_2$ if $\sum_{R_i\in t_1} 2^{i-1} < \sum_{R_i\in t_2}
  2^{i-1}$ for $t_1, t_2 \in T_{\tau}$.
   The \emph{type histogram} of the
  structure $\mathcal{A}$ is then defined as $\operatorname{hist}(\mathcal{A})
  \coloneq \left(e_t\right)_{t \in T_\tau}$, where $e_t$ is the number of
  elements in $\mathcal{A}$ that have type $t$.

  A \emph{type modulator} is a set $S_T$ consisting of triples $(t_1, t_2, a)$,
  which states that we modify element $a$ from having type $t_1$ to having type $t_2$. A
  type modulator is \emph{admissible for a structure $\mathcal{A}$}, if for
  every tuple $(t_1, t_2, a)$, the element $a$ has type $t_1$ in $\mathcal{A}$ and every element
  occurs in at most one triple of the type modulator. Note that each admissible type modulator $S_T$
  gives rise to a modulator $S$ and vice versa. Thus, we let $\|S_T\| \coloneq
  \|S\|$ and write $\mathcal{A} \triangleright S_T$ for the structure obtained by applying the modulator $S$ that corresponds to the type modulator
  $S_T$ to $\mathcal{A}$. We can define a lexicographic order on type modulators
  by letting $(t_1, t_2, a) <_{\text{lex}} (t_1', t_2', a')$ if (1) $t_1
  <_{\text{lex}} t_1'$ or if (2) $t_1 = t_1'$ and $t_2 <_{\text{lex}} t_2'$ or if
   (3) $t_1 = t_1'$, $t_2 = t_2'$, and $a < a'$, where $a < a'$ is the order we assume on the
  universe. Two type modulators $S_T, S_T'$ are \emph{congruent}, if for every
  $(t_1, t_2, a) \in S_T$ and $(t_1', t_2', a')\in S_T'$ at the same position in
  the order, we have that $t_1 = t_1'$ and $t_2 = t_2'$. Now, for two congruent
  type modulators $S_T$ and $S_T'$, we have $S_T <_{\text{lex}} S_T'$ if $a <
  a'$ for all pairs of tuples $(t_1, t_2, a) \in S_T$ and $(t_1', t_2', a')\in
  S_T'$ at the same position in the order.

  Note that two monadic structures are isomorphic if they have the same type
  histogram. This in turn means that two congruent type modulators yield the
  same structure up to isomorphism, so it suffices to consider
  lexicographically minimal type modulators. There are only
  $k^{|T_\tau|^2} = k^{2^{2r}}$ lexicographically minimal type modulators of size $k$, so the number of candidate modulators that need to be considered is polynomial in $k$.

  Now, with a $\Para\Class{AC}^0$-circuit, we can check in parallel, if for any of these candidates, it holds that $\mathcal{A} \triangleright
  S_T \models \phi$. If this is the case, we accept, otherwise, we reject.
\end{proof}

%%%%%%%%%%%%%%%%%%%%%%%%%%
% Classical
%%%%%%%%%%%%%%%%%%%%%%%%%%
\section{Classical Complexity of Relation Modification Problems}
\label{section:classical}
% \tcsautomoveaddto{main}{\subsection{Proofs for
% Section~\ref{section:classical}}}

The results from the previous section largely also apply to the
classical setting and help us to establish the following tractability frontier:
Relation modification problems are either 
solvable in $\Class{TC}^0$ or contain intractable problems.

\subsection{Undirected Graphs, Directed Graphs, and Arbitrary Structures}
As in the parameterized setting, we first provide
the tractability landscapes of relation modification problems for undirected and
directed graphs and arbitrary structures. Again, we find that these landscapes coincide and exhibit the following complexities.

% {Pattern-Based Complexity Dichotomy for $\EM{undir}{\otimes}(p)$ and $\EM{dir}{\otimes}(p)$}
\begin{theorem}\label{theorem:undir}
  Let $p \in \{a,e\}^*$ be a pattern. For each modification operation $\otimes\in \{\mathrm{del},
  \mathrm{add}, \mathrm{edit}\}$, we have
  \begin{enumerate}
  \item $\EM{arb}{\otimes}(p) \subseteq \Class{AC}^0$, if $p
    \preceq e^*$.
  \item $\EM{arb}{\otimes}(p) \subseteq \Class{TC}^0$, if
  $ea\preceq p
    \preceq e^*a$.
  \item $\EM{undir}{\otimes}(p)$ and $\EM{dir}{\otimes}(p)$ both contain an $\Class{NP}$-hard
    problem, if $ae \preceq p$ or $aa \preceq p$.
  \end{enumerate}
\end{theorem}

\subsubsection*{Upper Bounds.} We show
that the only two tractable fragments for unrestricted inputs can be solved even
by $\Class{AC}^0$-circuits or $\Class{TC}^0$-circuits, respectively.

\begin{lemma}\label{lemma:undir-e*}
  For each $\otimes\in \{\mathrm{del}, \mathrm{add}, \mathrm{edit}\}$,
  $\EM{arb}{\otimes}(e^*) \subseteq \Class{AC}^0$.
\end{lemma}

\begin{proof}
  Let $\tau = \{R_1^{a_1}, \dots, R_r^{a_r}\}$ and let $\phi$ be a $\tau$-formula
  of the form $\exists x_1 \cdots \exists x_c \psi$, where $\psi$ is
  quantifier-free. Given a structure $\mathcal{A} = (A, R_1^{\mathcal{A}},
  \dots, R_r^{\mathcal{A}})$, an $\Class{AC}^0$-circuit can consider all $|A|^c$
  assignments to the $x_i$ in parallel. For each assignment, we can go through
  the $2^{\sum_{i=1}^{r} c^{a_i}}$ possible interpretations of relations, and
  for each interpretation that satisfies the formula, determine the distance~$d$ to
  the structure induced by the $x_i$. Since $d$ is a constant, we
  can check with an $\Class{AC}^0$-circuit if $d\leq k$ holds for one of the
  modifications to one of the assignments, and, if this is the case, accept.
\end{proof}

\begin{lemma}\label{lemma:undir-e*a}
  For each $\otimes\in \{\mathrm{del}, \mathrm{add}, \mathrm{edit}\}$,
  $\EM{arb}{\otimes}(e^*a) \subseteq \Class{TC}^0$.
\end{lemma}

\begin{proof}
  We prove the upper bound for the general case ${\otimes} =
  \mathrm{edit}$; more special cases follow by restricting 
  all editing operations to additions or deletions only and adapting the way we count modifications.
  
  Let $\phi$ be of the form $\exists x_1\exists x_2\cdots \exists x_c \forall y\
  (\psi)$, where $\psi$ is quantifier-free. We construct a $\Class{TC}^0$-circuit that on input $(\mathcal{A},
  k)$ correctly decides whether we can edit at most $k$ tuples
  such that for the resulting structure $\mathcal{A}^\dagger$, we have
  $\mathcal{A}^\dagger\models \phi$. For this, we consider all $|A|^c$ possible
  assignments to the~$x_i$ in parallel. We call the set $A'$ of elements in such an
  assignment a \emph{certificate} and denote the substructure induced by $A'$ as~$\mathcal{A}'$. For each certificate $A'$, we consider in parallel, all
  $2^{\sum_{i = 1}^{r} c^{a_i}}$ possible structures~$\mathcal{A}^* = (A',
  R_1^{\mathcal{A}^*}, \dots, R_r^{\mathcal{A}^*})$ defined on $A'$. Let
  $S$ denote the modulator that turns $\mathcal{A}'$ into $\mathcal{A}^*$. Next,
  once more in parallel, consider each element $v\in A\setminus A'$. There are up to $2^{\sum_{i=1}^{r}((c+1)^{a_i}) - c^{a_i}}$ possible ways of how~$v$
  interacts with the elements in~$A'$ and thus as many possible structures
  $\mathcal{A}_v = ((A' \times v)\cup (v\times A'), R_1^{\mathcal{A}_v}, \dots,
  R_r^{\mathcal{A}_v})$ with corresponding modulators $S_v$ that turn the substructure
  induced by $A'$ and $v$ into $\mathcal{A}_v$. Now, for
  some~$R_i^{\mathcal{A}_v}$, the formula~$\psi$ may hold when we interpret $y$
  by~$v$ and interpret the relations as $R_i^{\mathcal{A}^*}$ on~$A'$ and as
  $R_i^{\mathcal{A}_v}$ if they contain elements from $A'$ and~$v$. If there is such a
  structure~$\mathcal{A}_v$, pick one for which $\|S_v\|$ is minimal and note that, crucially,
  \emph{for a fixed certificate~$A'$, all~$S_v$ are disjoint and can be computed
  independently.} This implies that we can accept $(\mathcal{A},k)$ if and only if for some
  certificate $A'$, we have $\|S\| + \sum_{v\in A\setminus A'} \|S_v\|\leq k$.
  This computation can be performed by a $\Class{TC}^0$-circuit using majority
  gates.
\end{proof}

\subsubsection*{Lower Bounds.} 

We start by showing that the fragment considered in Lemma~\ref{lemma:undir-e*a} is not contained in $\Class{AC}^0$. We
will do so by reducing from the $\Lang{Majority}$ problem:
\begin{problem}
\problemtitle{$\Lang{Majority}$}
  \probleminput{A bitstring $s$.}
  \problemquestion{ Are at least half of the bits in $s$ set to $1$?} 
\end{problem}
The problem $\Lang{Majority}$ is known to not be contained in $\Class{AC}^0$~\cite[Chapter 3]{Vollmer99}.

\begin{lemma}\label{lemma:undir-a}
  For each $\otimes\in \{\mathrm{del}, \mathrm{add}, \mathrm{edit}\}$,
  $\EM{undir}{\otimes}(a) \not\subseteq \Class{AC}^0$.
\end{lemma}

\begin{proof}
  Consider the formula $\phi = \forall x (\neg x\adj x)$. We reduce from the problem
  $\Lang{Majority}$ to $\EM{undir}{\mathrm{del}}(\phi)$ as follows: We create a
  graph $G = (V, E)$ from the given bitstring $s$ by adding a vertex
  $v$ with $(v, v) \in E$ for each bit that is set to $0$ in $s$ and set $k = |s|/2$. Clearly, we can make the graph
  satisfy the formula if and only if at least half of the bits in $s$ were set
  to~$1$.
\end{proof}

Since we have shown in Lemma~\ref{lemma:undir-e*a} that all problems expressible
with the quantifier pattern $e^*a$ are tractable, natural candidates to study next are the patterns $aa$ and $ae$.
We will show that both suffice to express
$\Class{NP}$-hard problems.

\begin{lemma}\label{lemma:undir-aa}
  For each modification operation $\otimes\in \{\mathrm{del}, \mathrm{add}, \mathrm{edit}\}$,
  $\EM{undir}{\otimes}(aa)$ and $\EM{dir}{\otimes}(aa)$ each contain an
  $\Class{NP}$-hard problem.
\end{lemma}

\begin{proof}
  We can reduce $\Lang{Vertex-Cover}$ to 
  $\EM{undir}{\mathrm{add}}(\phi_{\ref{lemma:undir-aa}})$. Here,
  $\Lang{Vertex-Cover}$ is the $\Class{NP}$-complete problem that given a basic
  graph
  $(V,E,k)$ asks whether there is a set $C \subseteq V$ with $|C| \le k$
  such that for all $(u,v) \in E$ we have $\{u,v\} \cap C \neq
  \emptyset$; and we define
  $
    \phi_{\ref{lemma:undir-aa}} = \forall x_1 \forall x_2 \penalty-100
    (x_1 \adj x_2 \to (x_1 \adj x_1
    \lor x_2 \adj x_2))
  $,
  which expresses that every edge is incident to at least one vertex with a
  self-loop. To see that this reduction also works for edge  editing, note that
  we can turn any solution that deletes some edge $(u,v)$ into a
  non-larger solution that adds a self-loop at~$v$ instead. 
%
%%   \emph{The reduction.} We reduce from the $\Class{NP}$-complete problem
%%   $\Lang{vertex-cover}$ \emph{for basic graphs} to the edge completion problem. Let $(V, E, k)$ be
%%   given as
%%   input. The reduction is simply the identity function.
%
%%   \emph{Forward direction.}
%%   Let  $(V, E, k) \in \Lang{vertex-cover}$ be given. We need
%%   to show that $(G', k') \in \EM{undir}{\mathrm{add}}(\phi_{\ref{lemma:undir-aa}})$
%%   holds. Let $C \subseteq V$ with $|C| \le k$ be a vertex cover. Then add the
%%   edges $E' = \{(v, v)\mid v\in C\}$. Since $C$ is a vertex cover of size at most
%%   $k$, we have $|E'|\leq k$ and every edge has at least one vertex with a self-loop.
%
%%   \emph{Backward direction.}
%%   Conversely, suppose that for $G' = (V', E')$ we are given a set $A \subseteq
%%   V'\times V'$ with $\|A\| \le k$ such that $(V', E'\cup A) \models
%%   \phi_{\ref{lemma:undir-aa}}$. Take $C = \{v \mid (v, v) \in A\}$. Then, $C$ is a vertex cover of size at
%%   most $k$ of $G$, since by the formula, we have to add a loop to at least one vertex of every
%%   edge in the graph.
% 
%%   \emph{Edge editing.}
%%   To see that the reduction also works for the edge editing problem, notice that
%%   we cannot add self-loops by deleting edges. Deleting an edge does not yield a
%%   more cost-effective strategy, since instead of deleting an edge $(v, u) \in
%%   G'$, which solves only one violation we can always add the self-loop $(v, v)$,
%%   fixing $\operatorname{deg}(v) \geq 1$ violations, where
%%   $\operatorname{deg}(v)$ is the degree of $v$.
\end{proof}

\begin{lemma}\label{lemma:undir-ae}
  For each modification operation $\otimes\in \{\mathrm{del}, \mathrm{add}, \mathrm{edit}\}$,
  $\EM{undir}{\otimes}(ae)$ contains an $\Class{NP}$-hard problem. 
\end{lemma}

\begin{proof}
  This follows directly from the proof of Lemma~\ref{lemma:p-undir-ae}.
\end{proof}

\subsection{Basic Graphs}
% {Pattern-Based Complexity Dichotomy for $\EM{basic}{\otimes}(p)$}
Finally, we consider the restriction of input and target structures to basic graphs and prove the following tractability frontier.
\begin{theorem}\label{theorem:basic}
  Let $p \in \{a,e\}^*$ be a pattern. For each modification operation $\otimes\in \{\mathrm{del},
  \mathrm{add}, \mathrm{edit}\}$, we have
  \begin{enumerate}
  \item $\EM{basic}{\otimes}(p) \subseteq \Class{AC}^0$, if $p \preceq e^*$ or $p
    \preceq a$.
  \item $\EM{basic}{\otimes}(p) \not\subseteq \Class{AC}^0$, but
  $\EM{basic}{\otimes}(p) \subseteq \Class{TC}^0$, if $aa \preceq p$ or $ae
  \preceq p$ or $ea \preceq p$ and $p \preceq e^*a$ or $p
    \preceq aa$ or $p \preceq ae$.
  \item $\EM{basic}{\otimes}(p)$ and $\EM{dir}{\otimes, }(p)$ each contain an $\Class{NP}$-hard
    problem, if $aaa\preceq p$ or $aea\preceq p$ or $aee\preceq p$ or $aae\preceq p$ or $eae\preceq p$ or $eaa \preceq p$.
  \end{enumerate}
\end{theorem}

\subsubsection*{Upper Bounds.} As in the parameterized setting, some fragments
become trivial for basic graphs. The main contributor to this fact is that some
terms, for example $x\adj x$, can now be trivially evaluated to false and the problems
defined by the remaining formulas are easy to solve. In the following, we analyze the tractable fragments individually.

\begin{lemma}\label{lemma:basic-a}
  For each $\otimes\in \{\mathrm{del}, \mathrm{add}, \mathrm{edit}\}$,
  $\EM{basic}{\otimes}(e^*) \subseteq \Class{AC}^0$ and $\EM{basic}{\otimes}(a)
  \subseteq \Class{AC}^0$.
\end{lemma}

\begin{proof}
  The first statement follows from Lemma~\ref{lemma:undir-e*}. For the second
  statement, since we do not allow self-loops in basic graphs, all formulas that
  only use one universal quantifier are either trivially true
  or false.
\end{proof}

\begin{lemma}\label{lemma:basic-ae}
  For each $\otimes\in \{\mathrm{del}, \mathrm{add}, \mathrm{edit}\}$, $\EM{basic}{\otimes}(ae) \subseteq \Class{TC}^0$.
\end{lemma}

\begin{proof}
  We can reuse the argument from the proof of Lemma~\ref{lemma:p-basic-ae}, where we showed that the parameterized versions of the fragment lie in
  $\Para\Class{AC}^0$. The key observation was that for
  most formulas, deciding the edge editing problem is trivial; and for the non-trivial cases we only needed the power of $\Para\Class{AC}^0$ to
  count the number of universal and\,/\,or isolated vertices. In the non-parameterized setting, this counting can be done
  using $\Class{TC}^0$-circuits.
\end{proof}

\begin{lemma}\label{lemma:basic-aa}
  For each $\otimes\in \{\mathrm{del}, \mathrm{add}, \mathrm{edit}\}$, $\EM{basic}{\otimes}(aa) \subseteq \Class{TC}^0$.
\end{lemma}

\begin{proof}
  Let $\phi$ be of the form $\forall x\forall y (\psi)$, where $\psi$ is
  quantifier-free. Below, we prove the statement for edge editing -- our arguments then transfer to the
  deletion and addition versions by allowing only the
  respective operations (this will be made more specific later). Again, we may assume that the input graph has at
  least two vertices and that $\psi$ is in disjunctive normal form and may contain the following disjuncts, without loss of generality in the order that we list them below:
  \vspace{-0.2cm}
  \begin{multicols}{2}
  \begin{enumerate}
  \item $x \neq  y$,
  \item $x = y$,
  \item $x \adj y$,
  \item $x \nadj y$,
  \item $x \neq y \land x\adj y$, and
  \item $x \neq y \land x\nadj y$.
  \end{enumerate}
  \end{multicols}
    \vspace{-0.2cm}
  In each of the following paragraphs, we argue that it holds that
  $\EM{basic}{\otimes}(\forall x\forall y (\psi)) \in \Class{TC}^0$ when the
  first disjunct in~$\psi$ is a particular one of the above.
  
  Suppose $x \neq y$ is the first disjunct of~$\psi$. Then $\forall
  x \forall y (\psi)$ is automatically true for all choices of
  distinct vertices that get assigned to~$x$ and~$y$. When the same vertex is assigned to $x$ and $y$,
  the formula will evaluate to true if the disjunct $x \nadj y$ or
  $x=y$ is present; otherwise it will evaluate to false and no modification
  of a basic graph can change this fact.

  Suppose $x = y$ is the first disjunct of~$\psi$. Observe that $x =
  y \lor (x \neq y \land x\adj y) \equiv\penalty-100 x =y \lor x \adj y$ and also
  $x = y \lor (x \neq y \land x\nadj y) \equiv x =y \lor x \nadj y$,
  so $\psi$ is equivalent to a disjunction of $x=y$ and a subset of the
  two literals $x\adj y$ and $x \nadj y$. Clearly, if $\psi$ is just
  $x=y$, the formula $\phi$ is trivially false (on graphs with at least two
  vertices); and if $\psi$ is $x =y \lor x\adj y \lor x \nadj y$, it
  is trivially true. The remaining cases are $x=y \lor x \adj y$ and
  $x=y \lor x\nadj y$. The first formula is satisfied if the input graph is a clique
  (without self-loops) or can be turned into a clique by adding $k$ edges (in the deletion case, this operation is of course not allowed). Clearly, counting the number of missing edges can be done in $\Class{TC}^0$. The second formula is satisfied if the input graph is edge-less or can be made edge-less by deleting up to $k$ edges (in the addition case, this operation is of course not allowed). Clearly, checking whether there are at most $k$ edges in the input graph can also be done in $\Class{TC}^0$.

  Suppose $x \adj y$ is the first disjunct of~$\psi$. If it is the
  only disjunct, then $\phi$ is false (no self-loops are allowed) and this
  cannot be fixed by modifying the graph. If $x\nadj y$ is also
  present, $\phi$ is trivially true. The disjunct $x\neq y \land x
  \adj y$ logically implies $x\adj y$ and can thus be removed. The
  remaining case is $\psi = x\adj y \lor (x \neq y \land x\nadj y)\equiv x\adj y
  \lor x\neq y$
  and this is always false in basic graphs that contain at least two vertices.

  Suppose $x \nadj y$ is the first disjunct of~$\psi$. We can then
  remove $x \neq y \land x\nadj y$, if present, as it logically
  implies $x \nadj y$. This leaves $\psi = x \nadj y$, which is once
  more the case that the target graph must be edge-less, meaning that at
  most $k$ edges may be present in the input graph, and $\psi = x
  \nadj y \lor (x \neq y \land x\adj y)$, which is always true in
  basic graphs.

  The remaining cases are that $\psi$ is some combination of
  disjunctions of $x \neq y \land x\adj y$ and $x \neq y \land x\nadj
  y$, but neither can be made true when $x$ and $y$ are interpreted
  identically. Thus, for such combinations, $\phi$ always evaluates to false.
\end{proof}

\subsubsection*{Lower Bounds.}

We again start by showing the circuit lower bounds.

\begin{lemma}[$\star$]\label{lemma:basic-ae-aa-ea}
  For each $\otimes\in \{\mathrm{del}, \mathrm{add}, \mathrm{edit}\}$, we have that
  $\EM{undir}{\otimes}(ae),$ $\EM{undir}{\otimes}(ea), \EM{undir}{\otimes}(aa) \not\subseteq \Class{AC}^0$.
\end{lemma}

In the previous section, we have shown that for basic graphs, the quantifier patterns $e^*a$, $ae$,
and $aa$ are tractable. Now, we turn our attention to the remaining fragments
and show that they contain $\Class{NP}$-hard problems. For some fragments, we
can reuse reductions that we provided in Section~\ref{section:parameterized}.

\begin{lemma}\label{lemma:basic-aea}
  For each modification operation $\otimes\in \{\mathrm{del}, \mathrm{add},
  \mathrm{edit}\}$, $\EM{basic}{\otimes}(aea)$, $\EM{basic}{\otimes}(aee)$,
  $\EM{basic}{\otimes}(aae)$ and $\EM{basic}{\otimes}(eae)$
  each contain an $\Class{NP}$-hard problem.
\end{lemma}

\begin{proof}
  Observe that the reductions from $\PLang[]{Set-Cover}$ to problems in
  $\PEM{basic}{\otimes}(aea)$ (Lemma~\ref{lemma:p-basic-aea}),
  $\PEM{basic}{\otimes}(aee)$ (Lemma~\ref{lemma:p-basic-aee}),
  $\PEM{basic}{\otimes}(aae)$ (Lemma~\ref{lemma:p-basic-aae}), and
  $\PEM{basic}{\otimes}(eae)$ (Lemma~\ref{lemma:p-basic-eae}) work analogously for
  the non-parameterized case, yielding the statement.
\end{proof}

The only cases left to consider are the patterns $eaa$ and $aaa$.
We will show in the following that in both quantifier patterns suffice to express $\Class{NP}$-hard problems.

\begin{lemma}\label{lemma:basic-eaa} For each modification operation $\otimes\in \{\mathrm{del},
  \mathrm{add}, \mathrm{edit}\}$, $\EM{basic}{\otimes}(eaa)$ and
  $\EM{dir}{\otimes}(eaa)$ each contain an
  $\Class{NP}$-hard problem.
\end{lemma}

\begin{proof}
  \emph{The formula.} Consider the formula
  \begin{align*}
    \phi_{\ref{lemma:basic-eaa}} = \exists c \forall x \forall y ((x\adj y \land x\neq c \land y\neq c) \to (x\nadj c \lor y\nadj c)), %\label{lemma:basic-eaa-eq}
  \end{align*}
  which states that there is a vertex that is not adjacent to both endpoints of any edge.

  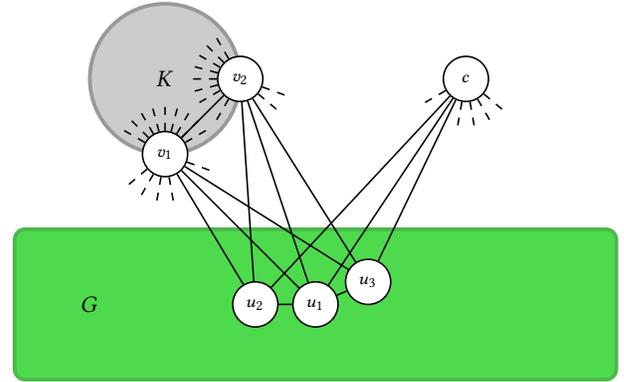
\begin{figure}
    \centering
    \begin{tikzpicture}[node/.append style={minimum size=6mm}]

      % G
      \draw[rounded corners, line width=.5mm, green!60!black,
      fill=green!80!black, opacity=0.7] (0, 0)
      rectangle (8, 2);
      \node (labelG) at (1, 1) {$G$};

      % c
      \node (c) at (6, 4) [node] {$c$};

      % K
      \draw[line width=.5mm, black!40, fill=black!20] (2, 4) circle (1cm);
      \node (labelK) at (2, 4) {$K$};

      % Vertices
      \node (cliqueNode) at (2, 3) [node] {$v_{1}$};
      \node (cliqueNode1) at (3, 4) [node] {$v_{2}$};

      \node (graphNode1) at (4, 1) [node] {$u_{1}$};
      \node (graphNode2) at (3.2, 1) [node] {$u_{2}$};
      \node (graphNode3) at (4.7, 1.3) [node] {$u_{3}$};

      % Solid Edges
      \draw (graphNode1) -- (graphNode2);
      \draw (graphNode1) -- (graphNode3);
      \draw (c) -- (graphNode1);
      \draw (c) -- (graphNode2);
      \draw (c) -- (graphNode3);
      \draw (cliqueNode) -- (cliqueNode1);
      \draw (cliqueNode) -- (graphNode1);
      \draw (cliqueNode) -- (graphNode2);
      \draw (cliqueNode) -- (graphNode3);
      \draw (cliqueNode1) -- (graphNode1);
      \draw (cliqueNode1) -- (graphNode2);
      \draw (cliqueNode1) -- (graphNode3);

      % Indicate other edges
      % Clique to clique
      \draw[dashed] (cliqueNode) -- ++(30:0.7);
      \draw[dashed] (cliqueNode) -- ++(45:0.7);
      \draw[dashed] (cliqueNode) -- ++(60:0.7);
      \draw[dashed] (cliqueNode) -- ++(75:0.7);
      \draw[dashed] (cliqueNode) -- ++(90:0.7);
      \draw[dashed] (cliqueNode) -- ++(105:0.7);
      \draw[dashed] (cliqueNode) -- ++(120:0.7);
      \draw[dashed] (cliqueNode) -- ++(135:0.7);
      \draw[dashed] (cliqueNode) -- ++(150:0.7);

      \draw[dashed] (cliqueNode1) -- ++(120:0.7);
      \draw[dashed] (cliqueNode1) -- ++(135:0.7);
      \draw[dashed] (cliqueNode1) -- ++(150:0.7);
      \draw[dashed] (cliqueNode1) -- ++(165:0.7);
      \draw[dashed] (cliqueNode1) -- ++(180:0.7);
      \draw[dashed] (cliqueNode1) -- ++(195:0.7);
      \draw[dashed] (cliqueNode1) -- ++(210:0.7);
      \draw[dashed] (cliqueNode1) -- ++(240:0.7);

      % Clique to G
      \draw[dashed] (cliqueNode) -- ++(220:0.7);
      \draw[dashed] (cliqueNode) -- ++(240:0.7);
      \draw[dashed] (cliqueNode) -- ++(260:0.7);
      \draw[dashed] (cliqueNode) -- ++(280:0.7);
      \draw[dashed] (cliqueNode) -- ++(340:0.7);

      \draw[dashed] (cliqueNode1) -- ++(320:0.7);
      \draw[dashed] (cliqueNode1) -- ++(340:0.7);

      % c to G
      \draw[dashed] (c) -- ++(210:0.7);
      \draw[dashed] (c) -- ++(260:0.7);
      \draw[dashed] (c) -- ++(280:0.7);
      \draw[dashed] (c) -- ++(300:0.7);
      \draw[dashed] (c) -- ++(320:0.7);

    \end{tikzpicture}
    \caption{Example for the reduction in
    Lemma~\ref{lemma:basic-eaa}.}
    \label{figure:basic-eaa}
  \end{figure}

  \emph{The reduction.} We reduce the problem $\Lang{Vertex-Cover}$ to the
  edge deletion problem as follows: On input $(G,
  k)$ with $G = (V, E)$, we construct $(G', k)$ with $G' = (V', E')$ by starting
  with $G$ and adding ($1$) a single vertex $c$ that we
  connect to every vertex in $V$ and ($2$) adding a clique $K$ of size $k + 2$,
  where we also connect every vertex in $K$ to every vertex in $V$. We claim
  that $(G, k)\in \Lang{Vertex-Cover}$ if and only if $(G', k)\in
  \EM{basic}{\mathrm{del}}(\phi_{\ref{lemma:basic-eaa}})$.
  
  An example is shown in Figure~\ref{figure:basic-eaa}. For a graph $G$ and number $k\in
  \mathbb{N}$, we display $u_1, u_2$, and $u_3$ with edges $u_1\undiradj' u_2$
  and $u_1\undiradj' u_3$ explicitly. We add a vertex $c$ which we connect to all vertices in~$G$. The edges to the $u_i$ are given as solid lines, while the edges to the
  other vertices in $G$ are indicated as dashed lines. We also add a clique $K$ of
  size $k + 2$, with $v_1$ and $v_2$ drawn explicitly as exemplary
  vertices. They are also connected to every vertex in~$G$.
  
  \emph{Forward direction.}
  Let  $(V, E, k) \in \Lang{Vertex-Cover}$ be given. We need
  to show that $(G', k) \in \EM{basic}{\mathrm{del}}(\phi_{\ref{lemma:basic-eaa}})$
  holds. Let $C \subseteq V$ with $|C| \le k$ be a vertex cover for $G$. Then, delete the
  edges $D = \{(v, c)\mid v\in C\} \cup \{(c, v)\mid v\in C\}$. We have
  $\|D\| = |C|\leq k$. We argue that for the resulting graph
  $G'$, we have $G'\models \phi_{\ref{lemma:basic-eaa}}$: Choose the
  existentially-quantified variable to be assigned with~$c$. Then, each vertex in the
  clique~$K$ is not connected to~$c$ by construction, and since $C$ is a vertex
  cover of $G$, for every edge in~$G$, one end is not connected to~$c$.

  \emph{Backward direction.}
  Conversely, suppose that for $G' = (V', E')$ we are given a set $D \subseteq
  E'$ with $\|D\| \le k$ such that $(V', E'\setminus D) \models
  \phi_{\ref{lemma:basic-eaa}}$. First, assume that $c$ is chosen as the
  assignment to the existentially quantified variable. Then, the deletion of at
  most $k$ edges has to ensure in particular that every edge in $G$ has a vertex
  that is not connected to $c$. Because of this, it does not make sense to
  delete an edge in $G$, since this can only resolve one conflicting edge, while
  deleting an edge from one of the vertices to $c$ resolves the same conflict and additionally
  all conflicts arising from incident edges to the vertex. Formally,
  if $(u,v) \in D$ and $(v,u)\in D$ with $c \notin \{u,v\}$, then
  remove both edges from~$D$ and add $(v,c)$ and $(c,v)$ instead;
  this will not increase the size of~$D$ and will still satisfy the
  formula. Now, $C = \{v \mid (v, c)\in D\}$ is a vertex
  cover of $G$ of size at most $k$. 

  We conclude by arguing that we cannot choose the existentially-quantified variable
  differently. Choosing the variable to be one of the $v\in K$ is not
  viable, since we would have to disconnect every other vertex in the clique
  from $v$, leading to at least $k + 1$ deletions. The same holds for a vertex
  $v\in V$.

  \emph{Edge editing.} Since we want to disconnect a vertex from at least one end point of every edge, adding edges is of no use.
\end{proof}

\begin{lemma}\label{lemma:basic-aaa}
  For each modification operation $\otimes\in \{\mathrm{del}, \mathrm{add}, \mathrm{edit}\}$, $\EM{basic}{\otimes}(aaa)$ contains an $\Class{NP}$-hard problem.
\end{lemma}

\begin{proof}
  Consider
  $
    \phi_{\text{clusters}} = \forall x\forall
    y\forall z ((x \adj y \land y \adj z) \to x \adj z)
  $,
  which is satisfied by so-called cluster graphs, that is, graphs that do not contain an induced $P_3$. Therefore, 
  $\EM{basic}{\mathrm{del}}(\phi_{\text{clusters}})$ and
  $\EM{basic}{\mathrm{edit}}(\phi_{\text{clusters}})$ characterize the
  problems \Lang{Cluster-} \Lang{Deletion} and \Lang{Cluster-Editing}, both of which
  are $\Class{NP}$-hard~\cite{Shamir2004173}.
\end{proof}

\subsection{Monadic Structures}

We conclude by considering structures over monadic
vocabularies.

\begin{theorem}\label{theorem:monadic}
  Let $p \in \{a,e\}^*$ be a pattern. For each modification operation $\otimes\in \{\mathrm{del},
  \mathrm{add}, \mathrm{edit}\}$, we have
  \begin{enumerate}
    \item $\EM{mon}{\otimes}(p) \subseteq \Class{AC}^0$, if $p\preceq e^*$.
    \item $\EM{mon}{\otimes}(p) \not \subseteq \Class{AC}^0$, but $\EM{mon}{\otimes}(p)\subseteq \Class{TC}^0$, if $a\preceq p$.
  \end{enumerate}
\end{theorem}

\begin{proof}
  The first statement follows directly from Lemma~\ref{lemma:undir-e*}.
  
  For the lower bound in the second statement, consider the formula $\phi =
  \forall x \neg R(x)$. We reduce from $\Lang{Majority}$ to
  $\EM{undir}{\mathrm{del}}(\phi)$. For each bit in $s$ that is set to $0$, add
  an element $e$ with $e \in R$ and set $k = |s|/2$. Now, we can satisfy the
  formula by deleting elements of~$R$ if and only if at most $|s|/2$ bits were
  set to $0$ in $s$.

  The upper bound in the second statement follows from
  Theorem~\ref{theorem:p-monadic}. Notice that the number of modulators we
  consider is polynomial in $k$, and checking if a modulator is admissable boils
  down to checking whether we have up to $k$ elements of a given type in the
  universe, which we can do with $\Class{TC}^0$-circuits.
\end{proof}

\section{Conclusion and Discussion}
%== What we have done
We provided a full classification of the parameterized and classical complexity
of relation modification problems towards target properties that are expressible
in different fragments of first-order logic. In our analysis, we considered (1)
different modification operations, namely deleting relations, adding relations,
and editing relations and (2)~different input structures, namely monadic structures, basic graphs,
undirected graphs, directed graphs, and arbitrary structures. 

%The classification is based on the quantifier patterns of the formulas and sheds additional light on the complexityproperties that emerge from these patterns: 
While the tractability frontier coincides for all modification operations, it differs with respect to the type of structure that is
allowed as input. More specifically, we obtain a different complexity landscape
on monadic structures and basic graphs each, whereas modifying all other input structures is equally hard.

The classifications we established are ``strong'' separations in the
sense that all fragments either yield relation modification problems that
are quite easy to solve (that is, they lie in $\Para\Class{AC}^{0\uparrow}$ or
$\Class{TC}^0$) or that are intractable (that is, they contain $\Class W[2]$-hard or
$\Class{NP}$-hard problems).
See Table~\ref{table:summary} for a full summary of our results.
%% %== Directed graphs
%% Note also that the results presented in this paper easily transfer to the
%% setting where instead of measuring the modification steps by the number of
%% undirected pairs and loops $\|D\|$, we use the number of (directed) edges
%% $|D|$. For basic and undirected graphs, the results directly transfer, while for
%% directed graphs, we have to map the parameter $k$ to $2k$ in the
%% Lemma~\ref{lemma:p-basic-aea}, Lemma~\ref{lemma:p-basic-aee},
%% Lemma~\ref{lemma:p-basic-aae}, Lemma~\ref{lemma:p-basic-eae}, and
%% Lemma~\ref{lemma:basic-eaa}.

%\todo{We showed that editing all relations given in the
%input. By our proofs, it is easy to see that the complexity landscape is also the
%same for the \emph{uniform} variant of the problems, where we may edit only a
%single (at least binary) relation.}
%
%\todo{We also saw that in our settings, it never makes a difference if the
%relations are unordered or ordered. It would be interesting to exhibit settings
%in which the order of relations actually makes a difference.}
Upon closer examination of our proofs, we find that considering different restrictions and counting measures for our modification budget yields identical complexity landscapes. More specifically, we observe the following:

\begin{observation}
For each modification operation $\otimes \in \{\mathrm{del}, \mathrm{add},$ $\mathrm{edit}\}$ and type $\mathrm{T} \in \{\mathrm{arb}, \mathrm{dir},
\mathrm{undir}, \mathrm{basic}, \mathrm{mon}\}$,
$\EM{T}{\otimes}$ ($\PEM{T}{\otimes}$, respectively) exhibits the complexity landscape as outlined in Table~\ref{table:summary},
  \begin{itemize}
    \item when up to $k$ tuples (across all relations) can be modified, 
    \item when up to $k$ tuples of a single relation can be modified, and
    \item when up to $k$ tuples can be modified in every relation.
  \end{itemize}
  Moreover, we find that for $\mathrm{T} \in \{\mathrm{arb}, \mathrm{dir}\}$ and modification budget~$k$, the following problem versions are equally hard:
  \begin{itemize}
  \item $\EM{T}{\otimes}$ (or $\PEM{T}{\otimes}$, resp.), where we count the number of modified ordered tuples and compare it against~$k$ and 
  \item $\EM{T}{\otimes}$ (or $\PEM{T}{\otimes}$, resp.), where we count the number of modified unordered tuples and compare it against~$k$.
  \end{itemize}
\end{observation}
The latter raises the following question:
Are there other natural counting versions that are of interest -- for instance in database repair -- under which the complexity of $\EM{\mathrm{arb}}{\otimes}$ and $\EM{\mathrm{dir}}{\otimes}$ (of $\PEM{\mathrm{arb}}{\otimes}$ and $\PEM{\mathrm{dir}}{\otimes}$, resp.) diverges?

In this paper, we only considered relational structures. It would be interesting to extend our framework to vocabularies
that contain function symbols. There, even fundamental questions like \emph{What are sensible
notions of editing on functions?} or \emph{Which problems can we additionally define by allowing editing on functions?} are
yet to be explored.

%== Other edge modifications
Moreover, the relation modification operations we considered in this paper are not
exhaustive: Graph theory studies other edge
modification operations such as \emph{edge
contraction}, \emph{edge splitting}, or \emph{edge flipping}. The study of
edge flipping may be of particular interest as this operation can be used to define a variant of the feedback
arc set problem, where instead of making a graph acyclic by deleting edges, we
have to flip them. This problem was recently used to solve the
$\Lang{One-Sided-Crossing-Minimization}$ problem efficiently in a practical
setting~\cite{BannachCKW24}.

%== Graph types
%Another direction of research could be extending these results to more types of
%graphs: One could consider also directed graphs without self-loops or arbitrary logical
%structures, where ``edge-deletion'' would be interpreted as deleting tuples from a
%relation of arbitrary arity. Even for bounds like Lemma~\ref{lemma:undir-e*a}, the situation
%would be unclear when studying arbitrary logical structures. One could also
%examine the complexity landscape on graph classes with bounded parameters.

%== Stronger logics
Another interesting direction of research is to extend our results to
higher-order logics, such as monadic second-order logic or existential
second-order logic. While these logics allow to express more properties of
interest, such as acyclicity, which allows us to model the feedback arc set
problem as an edge deletion or edge flipping problem, even model checking can
become intractable there. In that regard, standard extensions of first-order
logic could also be of interest, as they may capture more properties, such as
transitive closures, fixed-point operators, or counting.

%== Cluster editing
Finally, a closer inspection of the fragments tractable in the parameterized setting is of
interest: It immediately follows from our classical complexity results that the
fragment $e^*a$ is in $\Para\Class{AC}^0$ for all graph classes and
modification operations.
With respect to the fragment~$a^*$ -- which for instance can be used to describe
the interesting meta-problem $\Lang{$H$-free Deletion}$ for fixed graphs $H$ --
we know that the respective relation modification problems lie in
$\Para\Class{AC}^{0\uparrow}$, but the question of whether they can be placed in
$\Para\Class{AC}^0$ as well is open (see also~\cite[page 52]{Bannach2020}).

%\begin{acks}
%To Robert, for the bagels and explaining CMYK and color spaces.
%\end{acks}

\bibliographystyle{ACM-Reference-Format}
\bibliography{main}

@book{EbbinghausF05,
  author    = {Heinz{-}Dieter Ebbinghaus and
               J{\"{o}}rg Flum},
  title     = {Finite Model Theory},
  edition   = {2nd},
  publisher = {Springer},
  year      = {2005},
  doi       = {10.1007/3-540-28788-4},
  isbn      = {978-3-540-28787-2},
}

@book{Cygan15,
  author    = {Marek Cygan and
                Fedor V. Fomin and
                {\L}ukasz Kowalik and
                Daniel Lokshtanov and
                Dániel Marx and
                Marcin Pilipczuk and
                Micha{\l} Pilipczuk and
                Saket Saurabh},
  title     = {Parameterized Algorithms},
  publisher = {Springer},
  year      = {2015},
  doi       = {10.1007/978-3-319-21275-3},
  isbn      = {978-3-319-21274-6},
}

@article{GottlobKS04,
  author    = {Georg Gottlob and
               Phokion G. Kolaitis and
               Thomas Schwentick},
  title     = {Existential Second-Order Logic Over Graphs: Charting the Tractability Frontier},
  journal   = {Journal of the ACM},
  volume    = {51},
  number    = {2},
  pages     = {312--362},
  year      = {2004},
  doi       = {10.1145/972639.972646},
}

@inproceedings{Tantau15,
  author    = {Till Tantau},
  title     = {Existential Second-order Logic over Graphs: {A} Complete Complexity-theoretic
               Classification},
  booktitle = {32nd International Symposium on Theoretical Aspects of Computer Science,
               {STACS} 2015},
  pages     = {703--715},
  year      = {2015},
  doi       = {10.4230/LIPIcs.STACS.2015.703},
}

@book{FlumG06,
  author    = {J{\"{o}}rg Flum and
               Martin Grohe},
  title     = {Parameterized Complexity Theory},
  publisher = {Springer},
  year      = {2006},
  doi       = {10.1007/3-540-29953-X},
  isbn      = {978-3-540-29952-3},
}

@inproceedings{BannachT18,
  author    = {Max Bannach and
               Till Tantau},
  title     = {Computing Kernels in Parallel: Lower and Upper Bounds},
  booktitle = {13th International Symposium on Parameterized and Exact Computation,
               {IPEC} 2018},
  pages     = {13:1--13:14},
  year      = {2018},
  doi       = {10.4230/LIPIcs.IPEC.2018.13},
}

@article{BannachT20,
  author    = {Max Bannach and
               Till Tantau},
  title     = {Computing Hitting Set Kernels By {AC}\({}^{\mbox{0}}\)-Circuits},
  journal   = {Theory of Computing Systems},
  volume    = {64},
  number    = {3},
  pages     = {374--399},
  year      = {2020},
  doi       = {10.1007/s00224-019-09941-z},
}

@article{FominGT20,
  author    = {Fedor V. Fomin and
               Petr A. Golovach and
               Dimitrios M. Thilikos},
  title     = {On the Parameterized Complexity of Graph Modification to First-Order
               Logic Properties},
  journal   = {Theory of Computing Systems},
  volume    = {64},
  number    = {2},
  pages     = {251--271},
  year      = {2020},
  url       = {https://doi.org/10.1007/s00224-019-09938-8},
  doi       = {10.1007/s00224-019-09938-8},
}

@inproceedings{FominGT22,
  author    = {Fedor V. Fomin and
               Petr A. Golovach and
               Dimitrios M. Thilikos},
  title     = {Parameterized Complexity of Elimination Distance to First-Order Logic
               Properties},
  booktitle = {{ACM} Transactions on Computational Logic},
  volume       = {23},
  number       = {3},
  pages        = {17:1--17:35},
  year         = {2022},
  doi          = {10.1145/3517129},
}

@article{EiterGG00,
  author    = {Thomas Eiter and
               Yuri Gurevich and
               Georg Gottlob},
  title     = {Existential second-order logic over strings},
  journal   = {Journal of the {ACM}},
  volume    = {47},
  number    = {1},
  pages     = {77--131},
  year      = {2000},
  doi       = {10.1145/331605.331609},
}

@article{LewisY80,
  author    = {John M. Lewis and
               Mihalis Yannakakis},
  title     = {The Node-Deletion Problem for Hereditary Properties is {NP}-Complete},
  journal   = {Journal of Computer and System Sciences},
  volume    = {20},
  number    = {2},
  pages     = {219--230},
  year      = {1980},
  doi       = {10.1016/0022-0000(80)90060-4},
}

@book{DowneyF99,
  author       = {Rodney G. Downey and
                  Michael R. Fellows},
  title        = {Parameterized Complexity},
  series       = {Monographs in Computer Science},
  publisher    = {Springer},
  year         = {1999},
  url          = {https://doi.org/10.1007/978-1-4612-0515-9},
  doi          = {10.1007/978-1-4612-0515-9},
  isbn         = {978-1-4612-6798-0},
}

@inproceedings{BannachCT23,
  author       = {Max Bannach and
                  Florian Chudigiewitsch and
                  Till Tantau},
  editor       = {Neeldhara Misra and
                  Magnus Wahlstr{\"{o}}m},
  title        = {Existential Second-Order Logic over Graphs: Parameterized Complexity},
  booktitle    = {18th International Symposium on Parameterized and Exact Computation,
                  {IPEC} 2023},
  series       = {LIPIcs},
  volume       = {285},
  pages        = {3:1--3:15},
  publisher    = {Schloss Dagstuhl -- Leibniz-Zentrum f{\"{u}}r Informatik},
  year         = {2023},
  url          = {https://doi.org/10.4230/LIPIcs.IPEC.2023.3},
  doi          = {10.4230/LIPICS.IPEC.2023.3},
  bibsource    = {dblp computer science bibliography, https://dblp.org}
}

@inproceedings{Yannakakis78,
  author = {Yannakakis, Mihalis},
  title = {Node-and edge-deletion {NP}-complete problems},
  year = {1978},
  isbn = {9781450374378},
  publisher = {Association for Computing Machinery},
  address = {New York, NY, USA},
  url = {https://doi.org/10.1145/800133.804355},
  doi = {10.1145/800133.804355},
  booktitle = {Proceedings of the Tenth Annual ACM Symposium on Theory of Computing},
  pages = {253--264},
  numpages = {12},
  location = {San Diego, California, USA},
  series = {STOC '78}
}

@article{CRESPELLE2023100556,
  title = {A survey of parameterized algorithms and the complexity of edge modification},
  journal = {Computer Science Review},
  volume = {48},
  pages = {100556},
  year = {2023},
  issn = {1574-0137},
  doi = {https://doi.org/10.1016/j.cosrev.2023.100556},
  author = {Christophe Crespelle and Pål Grønås Drange and Fedor V. Fomin and Petr Golovach},
}

@book{BorgerGG1997,
  author       = {Egon B{\"{o}}rger and
                  Erich Gr{\"{a}}del and
                  Yuri Gurevich},
  title        = {The Classical Decision Problem},
  series       = {Perspectives in Mathematical Logic},
  publisher    = {Springer},
  year         = {1997},
}

@inproceedings{BannachCT24,
  author       = {Max Bannach and
                  Florian Chudigiewitsch and
                  Till Tantau},
  editor       = {Rastislav Kr{\'{a}}lovic and
                  Anton{\'{\i}}n Kucera},
  title        = {On the Descriptive Complexity of Vertex Deletion Problems},
  booktitle    = {49th International Symposium on Mathematical Foundations of Computer
                  Science, {MFCS} 2024},
  series       = {LIPIcs},
  volume       = {306},
  pages        = {17:1--17:14},
  publisher    = {Schloss Dagstuhl -- Leibniz-Zentrum f{\"{u}}r Informatik},
  year         = {2024},
  url          = {https://doi.org/10.4230/LIPIcs.MFCS.2024.17},
  doi          = {10.4230/LIPICS.MFCS.2024.17},
}

@book{Vollmer99,
  author    = {Heribert Vollmer},
  title     = {Introduction to Circuit Complexity -- A Uniform Approach},
  publisher = {Springer},
  year      = {1999},
  isbn      = {978-3-540-64310-4},
  series    = {Texts in Theoretical Computer Science},
}

@article{Yannakakis81,
  author       = {Mihalis Yannakakis},
  title        = {Edge-Deletion Problems},
  journal      = {{SIAM} Journal on Computing},
  volume       = {10},
  number       = {2},
  pages        = {297--309},
  year         = {1981},
  url          = {https://doi.org/10.1137/0210021},
  doi          = {10.1137/0210021},
}

@article{BurzynBD06,
  author       = {Pablo Burzyn and
                  Flavia Bonomo and
                  Guillermo Dur{\'{a}}n},
  title        = {{NP}-completeness results for edge modification problems},
  journal      = {Discrete Applied Mathematics},
  volume       = {154},
  number       = {13},
  pages        = {1824--1844},
  year         = {2006},
  url          = {https://doi.org/10.1016/j.dam.2006.03.031},
  doi          = {10.1016/J.DAM.2006.03.031},
}

@article{Natanzon2001109,
  title = {Complexity classification of some edge modification problems},
  journal = {Discrete Applied Mathematics},
  volume = {113},
  number = {1},
  pages = {109--128},
  year = {2001},
  issn = {0166-218X},
  doi = {https://doi.org/10.1016/S0166-218X(00)00391-7},
  url = {https://www.sciencedirect.com/science/article/pii/S0166218X00003917},
  author = {Assaf Natanzon and Ron Shamir and Roded Sharan},
}

@phdthesis{Mancini08,
  title        = {Graph modification problems related to graph classes},
  author       = {Frederico Mancini},
  year         = 2008,
  month        = {May},
  address      = {Norway},
  school       = {University of Bergen}
}

@article{Shamir2004173,
  title = {Cluster graph modification problems},
  journal = {Discrete Applied Mathematics},
  volume = {144},
  number = {1},
  pages = {173--182},
  year = {2004},
  note = {Discrete Mathematics and Data Mining},
  issn = {0166-218X},
  doi = {https://doi.org/10.1016/j.dam.2004.01.007},
  url = {https://www.sciencedirect.com/science/article/pii/S0166218X04001957},
  author = {Ron Shamir and Roded Sharan and Dekel Tsur},
}

@phdthesis{Bannach2020,
  author = {Max Bannach},
  school = {University of Lübeck},
  address = {Germany},
  title = {Parallel Parameterized Algorithms},
  year = {2020}
}

@inproceedings{BannachCKW24,
  author       = {Max Bannach and
                  Florian Chudigiewitsch and
                  Kim{-}Manuel Klein and
                  Marcel Wien{\"{o}}bst},
  editor       = {{\'{E}}douard Bonnet and
                  Pawel Rzazewski},
  title        = {{PACE} Solver Description: UzL Exact Solver for One-Sided Crossing
                  Minimization},
  booktitle    = {19th International Symposium on Parameterized and Exact Computation,
                  {IPEC} 2024},
  series       = {LIPIcs},
  volume       = {321},
  pages        = {28:1--28:4},
  publisher    = {Schloss Dagstuhl -- Leibniz-Zentrum f{\"{u}}r Informatik},
  year         = {2024},
  url          = {https://doi.org/10.4230/LIPIcs.IPEC.2024.28},
  doi          = {10.4230/LIPICS.IPEC.2024.28},
}

@inproceedings{BannachST15,
  author       = {Max Bannach and
                  Christoph Stockhusen and
                  Till Tantau},
  editor       = {Thore Husfeldt and
                  Iyad A. Kanj},
  title        = {Fast Parallel Fixed-parameter Algorithms via Color Coding},
  booktitle    = {10th International Symposium on Parameterized and Exact Computation,
                  {IPEC} 2015},
  series       = {LIPIcs},
  volume       = {43},
  pages        = {224--235},
  publisher    = {Schloss Dagstuhl -- Leibniz-Zentrum f{\"{u}}r Informatik},
  year         = {2015},
  url          = {https://doi.org/10.4230/LIPIcs.IPEC.2015.224},
  doi          = {10.4230/LIPICS.IPEC.2015.224},
}

@article{BrugmannKM09,
  author       = {Daniel Br{\"{u}}gmann and
                  Christian Komusiewicz and
                  Hannes Moser},
  title        = {On Generating Triangle-Free Graphs},
  journal      = {Electronic Notes in Discrete Mathematics},
  volume       = {32},
  pages        = {51--58},
  year         = {2009},
  url          = {https://doi.org/10.1016/j.endm.2009.02.008},
  doi          = {10.1016/J.ENDM.2009.02.008},
}

@article{ErdosGT96,
  author       = {Paul Erd{\"{o}}s and
                  Tibor Gallai and
                  Zsolt Tuza},
  title        = {Covering and independence in triangle structures},
  journal      = {Discrete Mathematics},
  volume       = {150},
  number       = {1-3},
  pages        = {89--101},
  year         = {1996},
  url          = {https://doi.org/10.1016/0012-365X(95)00178-Y},
  doi          = {10.1016/0012-365X(95)00178-Y},
}

@article{BujtasDDGTY25,
  author       = {Csilla Bujt{\'{a}}s and
                  Akbar Davoodi and
                  Laihao Ding and
                  Ervin Gy{\"{o}}ri and
                  Zsolt Tuza and
                  Donglei Yang},
  title        = {Covering the edges of a graph with triangles},
  journal      = {Discrete Mathematics},
  volume       = {348},
  number       = {1},
  pages        = {114226},
  year         = {2025},
  url          = {https://doi.org/10.1016/j.disc.2024.114226},
  doi          = {10.1016/J.DISC.2024.114226},
}

@inproceedings{ShengX23,
  author       = {Zimo Sheng and
                  Mingyu Xiao},
  editor       = {Weili Wu and
                  Guangmo Tong},
  title        = {A Discharging Method: Improved Kernels for Edge Triangle Packing and
                  Covering},
  booktitle    = {Computing and Combinatorics -- 29th International Conference, {COCOON}   2023},
  series       = {Lecture Notes in Computer Science},
  volume       = {14423},
  pages        = {171--183},
  publisher    = {Springer},
  year         = {2023},
  url          = {https://doi.org/10.1007/978-3-031-49193-1_13},
  doi          = {10.1007/978-3-031-49193-1_13},
}

@article{HuffnerKMN10,
  author       = {Falk H{\"{u}}ffner and
                  Christian Komusiewicz and
                  Hannes Moser and
                  Rolf Niedermeier},
  title        = {Fixed-Parameter Algorithms for Cluster Vertex Deletion},
  journal      = {Theory of Computing Systems},
  volume       = {47},
  number       = {1},
  pages        = {196--217},
  year         = {2010},
  url          = {https://doi.org/10.1007/s00224-008-9150-x},
  doi          = {10.1007/S00224-008-9150-X},
}

@inproceedings{KonstantinidisP19,
  author       = {Athanasios L. Konstantinidis and
                  Charis Papadopoulos},
  editor       = {Peter Rossmanith and
                  Pinar Heggernes and
                  Joost{-}Pieter Katoen},
  title        = {Cluster Deletion on Interval Graphs and Split Related Graphs},
  booktitle    = {44th International Symposium on Mathematical Foundations of Computer
                  Science, {MFCS} 2019},
  series       = {LIPIcs},
  volume       = {138},
  pages        = {12:1--12:14},
  publisher    = {Schloss Dagstuhl -- Leibniz-Zentrum f{\"{u}}r Informatik},
  year         = {2019},
  url          = {https://doi.org/10.4230/LIPIcs.MFCS.2019.12},
  doi          = {10.4230/LIPICS.MFCS.2019.12},
}

@book{Kudelic22,
  author       = {Robert Kudelic},
  title        = {Feedback Arc Set -- {A} History of the Problem and Algorithms},
  series       = {Springer Briefs in Computer Science},
  publisher    = {Springer},
  year         = {2022},
  url          = {https://doi.org/10.1007/978-3-031-10515-9},
  doi          = {10.1007/978-3-031-10515-9},
  isbn         = {978-3-031-10514-2},
}

@article{GaoHN13,
  author       = {Yong Gao and
                  Donovan R. Hare and
                  James Nastos},
  title        = {The parametric complexity of graph diameter augmentation},
  journal      = {Discrete Applied Mathematics},
  volume       = {161},
  number       = {10-11},
  pages        = {1626--1631},
  year         = {2013},
  url          = {https://doi.org/10.1016/j.dam.2013.01.016},
  doi          = {10.1016/J.DAM.2013.01.016},
}

@inproceedings{GuoHN04,
  author       = {Jiong Guo and
                  Falk H{\"{u}}ffner and
                  Rolf Niedermeier},
  editor       = {Rodney G. Downey and
                  Michael R. Fellows and
                  Frank K. H. A. Dehne},
  title        = {A Structural View on Parameterizing Problems: Distance from Triviality},
  booktitle    = {Parameterized and Exact Computation, First International Workshop,
                  {IWPEC} 2004},
  series       = {Lecture Notes in Computer Science},
  volume       = {3162},
  pages        = {162--173},
  publisher    = {Springer},
  year         = {2004},
  doi          = {10.1007/978-3-540-28639-4_15},
}

@book{I98,
	author = {Neil Immerman},
	editor = {},
	publisher = {Springer Verlag},
	title = {Descriptive Complexity},
  booktitle={Graduate Texts in Computer Science},
  year={1998},}

@book{GradelKLMSVVW07,
  author    = {Erich Gr{\"{a}}del and
               Phokion G. Kolaitis and
               Leonid Libkin and
               Maarten Marx and
               Joel Spencer and
               Moshe Y. Vardi and
               Yde Venema and
               Scott Weinstein},
  title     = {Finite Model Theory and Its Applications},
  series    = {Texts in Theoretical Computer Science. An {EATCS} Series},
  publisher = {Springer},
  year      = {2007},
  url       = {https://doi.org/10.1007/3-540-68804-8},
  doi       = {10.1007/3-540-68804-8},
  isbn      = {978-3-540-00428-8},
}

@book{G17,
  author       = {Martin Grohe},
  title        = {Descriptive Complexity, Canonisation, and Definable Graph Structure
                  Theory},
  series       = {Lecture Notes in Logic},
  volume       = {47},
  publisher    = {Cambridge University Press},
  year         = {2017},
  doi          = {10.1017/9781139028868},
  isbn         = {9781139028868},
  bibsource    = {dblp computer science bibliography, https://dblp.org}
}

@article{BSY99,
  author       = {Amir Ben{-}Dor and
                  Ron Shamir and
                  Zohar Yakhini},
  title        = {Clustering Gene Expression Patterns},
  journal      = {Jornal of Computational Biology},
  volume       = {6},
  number       = {3/4},
  pages        = {281--297},
  year         = {1999},
  doi          = {10.1089/106652799318274},
  bibsource    = {dblp computer science bibliography, https://dblp.org}
}

@article{B59,
 author = {Seymour Benzer},
 journal = {Proceedings of the National Academy of Sciences of the United States of America},
 number = {11},
 pages = {1607--1620},
 publisher = {National Academy of Sciences},
 title = {On the Topology of the Genetic Fine Structure},
 urldate = {2025-03-12},
 volume = {45},
 year = {1959}
}

@inproceedings{RBIL02,
author = {Rizzi, Romeo and Bafna, Vineet and Istrail, Sorin and Lancia, Giuseppe},
title = {Practical Algorithms and Fixed-Parameter Tractability for the Single Individual SNP Haplotyping Problem},
year = {2002},
isbn = {3540442111},
publisher = {Springer-Verlag},
address = {Berlin, Heidelberg},
booktitle = {Proceedings of the Second International Workshop on Algorithms in Bioinformatics},
pages = {29–43},
numpages = {15},
series = {WABI '02}
}

@article{BBC04,
  author       = {Nikhil Bansal and
                  Avrim Blum and
                  Shuchi Chawla},
  title        = {Correlation Clustering},
  journal      = {Machine Learning},
  volume       = {56},
  number       = {1-3},
  pages        = {89--113},
  year         = {2004},
  doi          = {10.1023/B:MACH.0000033116.57574.95},
  bibsource    = {dblp computer science bibliography, https://dblp.org}
}

@article{BETT94,
  author       = {Giuseppe Di Battista and
                  Peter Eades and
                  Roberto Tamassia and
                  Ioannis G. Tollis},
  title        = {Algorithms for Drawing Graphs: an Annotated Bibliography},
  journal      = {Journal of Computational Geometry},
  volume       = {4},
  pages        = {235--282},
  year         = {1994},
  doi          = {10.1016/0925-7721(94)00014-X},
  bibsource    = {dblp computer science bibliography, https://dblp.org}
}

@article{CNR89,
  author       = {Hyeong{-}Ah Choi and
                  Kazuo Nakajima and
                  Chong S. Rim},
  title        = {Graph Bipartization and via Minimization},
  journal      = {{SIAM} Journal on Discrete Mathematics},
  volume       = {2},
  number       = {1},
  pages        = {38--47},
  year         = {1989},
  doi          = {10.1137/0402004},
  bibsource    = {dblp computer science bibliography, https://dblp.org}
}

@article{HH87,
title = {A review of graph theory application to the facilities layout problem},
journal = {Omega},
volume = {15},
number = {4},
pages = {291-300},
year = {1987},
issn = {0305-0483},
doi = {https://doi.org/10.1016/0305-0483(87)90017-X},
author = {Mohsen MD Hassan and Gary L Hogg},
}

@article{JM96,
  author       = {Michael J{\"{u}}nger and
                  Petra Mutzel},
  title        = {Maximum Planar Subgraphs and Nice Embeddings: Practical Layout Tools},
  journal      = {Algorithmica},
  volume       = {16},
  number       = {1},
  pages        = {33--59},
  year         = {1996},
  doi          = {10.1007/BF02086607},
  bibsource    = {dblp computer science bibliography, https://dblp.org}
}

@inproceedings{GS76,
  author       = {Georges Gardarin and
                  Stefano Spaccapietra},
  editor       = {G. M. Nijssen},
  title        = {Integrity of Data Bases: {A} General Lockout Algorithm with Deadlock
                  Avoidance},
  booktitle    = {Modelling in Data Base Management Systems, Proceeding of the {IFIP} Working Conference on Modelling in Data Base Management Systems, Freudenstadt,
                  Germany},
  pages        = {395--412},
  publisher    = {North-Holland},
  year         = {1976},
  bibsource    = {dblp computer science bibliography, https://dblp.org}
}

@inproceedings{FSM15,
  author       = {Fedor V. Fomin and
                  Saket Saurabh and
                  Neeldhara Misra},
  editor       = {Jianxin Wang and
                  Chee{-}Keng Yap},
  title        = {Graph Modification Problems: {A} Modern Perspective},
  booktitle    = {Frontiers in Algorithmics - 9th International Workshop, {FAW} 2015},
  series       = {Lecture Notes in Computer Science},
  volume       = {9130},
  pages        = {3--6},
  publisher    = {Springer},
  year         = {2015},
  doi          = {10.1007/978-3-319-19647-3\_1},
  bibsource    = {dblp computer science bibliography, https://dblp.org}
}

@article{XSQ24,
  author       = {Shaobin Xu and
                  Minghui Sun and
                  Jun Qin},
  title        = {A High-Scalability Graph Modification System for Large-Scale Networks},
  journal      = {Computer Graphics Forum},
  volume       = {43},
  number       = {6},
  year         = {2024},
  doi          = {10.1111/CGF.15191},
  bibsource    = {dblp computer science bibliography, https://dblp.org}
}

@inproceedings{LT08,
author = {Liu, Kun and Terzi, Evimaria},
title = {Towards identity anonymization on graphs},
year = {2008},
isbn = {9781605581026},
publisher = {Association for Computing Machinery},
address = {New York, NY, USA},
doi = {10.1145/1376616.1376629},
booktitle = {Proceedings of the 2008 ACM SIGMOD International Conference on Management of Data},
pages = {93–106},
numpages = {14},
location = {Vancouver, Canada},
series = {SIGMOD '08}
}

@inproceedings{BPLRK16,
  author       = {Peter W. Battaglia and
                  Razvan Pascanu and
                  Matthew Lai and
                  Danilo Jimenez Rezende and
                  Koray Kavukcuoglu},
  editor       = {Daniel D. Lee and
                  Masashi Sugiyama and
                  Ulrike von Luxburg and
                  Isabelle Guyon and
                  Roman Garnett},
  title        = {Interaction Networks for Learning about Objects, Relations and Physics},
  booktitle    = {Advances in Neural Information Processing Systems 29: Annual Conference
                  on Neural Information Processing Systems 2016},
  pages        = {4502--4510},
  year         = {2016},
  bibsource    = {dblp computer science bibliography, https://dblp.org}
}

@inproceedings{LSB12,
  author       = {Xuesong Lu and
                  Yi Song and
                  St{\'{e}}phane Bressan},
  editor       = {Stephen W. Liddle and
                  Klaus{-}Dieter Schewe and
                  A Min Tjoa and
                  Xiaofang Zhou},
  title        = {Fast Identity Anonymization on Graphs},
  booktitle    = {Database and Expert Systems Applications - 23rd International Conference,
                  {DEXA} 2012},
  series       = {Lecture Notes in Computer Science},
  volume       = {7446},
  pages        = {281--295},
  publisher    = {Springer},
  year         = {2012},
  doi          = {10.1007/978-3-642-32600-4\_21},
  bibsource    = {dblp computer science bibliography, https://dblp.org}
}

@misc{bodirsky2025,
      title={Hereditary First-Order Logic: the tractable quantifier prefix classes}, 
      author={Manuel Bodirsky and Santiago Guzmán-Pro},
      year={2025},
      eprint={2411.10860},
      archivePrefix={arXiv},
      primaryClass={math.LO},
      url={https://arxiv.org/abs/2411.10860}, 
}

@misc{bodirsky2025-1,
      title={On the Computational Power of Extensional ESO}, 
      author={Manuel Bodirsky and Santiago Guzmán Pro},
      year={2025},
      eprint={2511.08515},
      archivePrefix={arXiv},
      primaryClass={math.LO},
      url={https://arxiv.org/abs/2511.08515}, 
}

@inproceedings{MorelleST25,
  author       = {Laure Morelle and
                  Ignasi Sau and
                  Dimitrios M. Thilikos},
  editor       = {Anne Benoit and
                  Haim Kaplan and
                  Sebastian Wild and
                  Grzegorz Herman},
  title        = {Graph Modification of Bounded Size to Minor-Closed Classes as Fast
                  as Vertex Deletion},
  booktitle    = {33rd Annual European Symposium on Algorithms, {ESA} 2025},
  series       = {LIPIcs},
  volume       = {351},
  pages        = {7:1--7:18},
  publisher    = {Schloss Dagstuhl - Leibniz-Zentrum f{\"{u}}r Informatik},
  year         = {2025},
  url          = {https://doi.org/10.4230/LIPIcs.ESA.2025.7},
  doi          = {10.4230/LIPICS.ESA.2025.7},
  bibsource    = {dblp computer science bibliography, https://dblp.org}
}

@inproceedings{CesatiI98,
  author       = {Marco Cesati and
                  Miriam Di Ianni},
  editor       = {David J. Pritchard and
                  Jeff Reeve},
  title        = {Parameterized Parallel Complexity},
  booktitle    = {Euro-Par '98 Parallel Processing, 4th International Euro-Par Conference,
                  Southampton, UK, September 1--4, 1998, Proceedings},
  series       = {Lecture Notes in Computer Science},
  volume       = {1470},
  pages        = {892--896},
  publisher    = {Springer},
  year         = {1998},
  doi          = {10.1007/BFB0057945},
}

@article{ElberfeldST15,
  author       = {Michael Elberfeld and
                  Christoph Stockhusen and
                  Till Tantau},
  title        = {On the Space and Circuit Complexity of Parameterized Problems: Classes
                  and Completeness},
  journal      = {Algorithmica},
  volume       = {71},
  number       = {3},
  pages        = {661--701},
  year         = {2015},
  doi          = {10.1007/S00453-014-9944-Y},
}

@article{ChenF19,
  author       = {Yijia Chen and
                  J{\"{o}}rg Flum},
  title        = {Some lower bounds in parameterized AC\({}^{\mbox{0}}\)},
  journal      = {Information and Computation},
  volume       = {267},
  pages        = {116--134},
  year         = {2019},
  doi          = {10.1016/J.IC.2019.03.008},
}

@article{Abu-KhzamLSS06,
  author       = {Faisal N. Abu{-}Khzam and
                  Michael A. Langston and
                  Pushkar Shanbhag and
                  Christopher T. Symons},
  title        = {Scalable Parallel Algorithms for {FPT} Problems},
  journal      = {Algorithmica},
  volume       = {45},
  number       = {3},
  pages        = {269--284},
  year         = {2006},
  doi          = {10.1007/S00453-006-1214-1},
}

@article{CheethamDRST03,
  author       = {James Cheetham and
                  Frank K. H. A. Dehne and
                  Andrew Rau{-}Chaplin and
                  Ulrike Stege and
                  Peter J. Taillon},
  title        = {Solving large {FPT} problems on coarse-grained parallel machines},
  journal      = {Journal of Computing and System Sciences},
  volume       = {67},
  number       = {4},
  pages        = {691--706},
  year         = {2003},
}

@article{BannachST23,
  author       = {Max Bannach and
                  Malte Skambath and
                  Till Tantau},
  title        = {On the Parallel Parameterized Complexity of MaxSAT Variants},
  journal      = {Journal of Artificial Intelligence Research},
  volume       = {78},
  year         = {2023},
  doi          = {10.1613/JAIR.1.14748},
}

@inproceedings{BannachT16,
  author       = {Max Bannach and
                  Till Tantau},
  editor       = {Jiong Guo and
                  Danny Hermelin},
  title        = {Parallel Multivariate Meta-Theorems},
  booktitle    = {11th International Symposium on Parameterized and Exact Computation,
                  {IPEC} 2016},
  series       = {LIPIcs},
  volume       = {63},
  pages        = {4:1--4:17},
  publisher    = {Schloss Dagstuhl -- Leibniz-Zentrum f{\"{u}}r Informatik},
  year         = {2016},
  doi          = {10.4230/LIPICS.IPEC.2016.4},
}

@inproceedings{Hegeman0L25,
  author       = {Steef Hegeman and
                  Jan Martens and
                  Alfons Laarman},
  editor       = {Akanksha Agrawal and
                  Erik Jan van Leeuwen},
  title        = {Uniformity Within Parameterized Circuit Classes},
  booktitle    = {20th International Symposium on Parameterized and Exact Computation,
                  {IPEC} 2025},
  series       = {LIPIcs},
  volume       = {358},
  pages        = {27:1--27:16},
  publisher    = {Schloss Dagstuhl - Leibniz-Zentrum f{\"{u}}r Informatik},
  year         = {2025},
  doi          = {10.4230/LIPICS.IPEC.2025.27},
}

@inproceedings{Arenas99,
author = {Arenas, Marcelo and Bertossi, Leopoldo and Chomicki, Jan},
title = {Consistent query answers in inconsistent databases},
year = {1999},
isbn = {1581130627},
publisher = {Association for Computing Machinery},
address = {New York, NY, USA},
url = {https://doi.org/10.1145/303976.303983},
doi = {10.1145/303976.303983},
booktitle = {Proceedings of the Eighteenth ACM SIGMOD-SIGACT-SIGART Symposium on Principles of Database Systems},
pages = {68--79},
numpages = {12},
location = {Philadelphia, Pennsylvania, USA},
series = {PODS '99}
}

@article{CHOMICKI200590,
title = {Minimal-change integrity maintenance using tuple deletions},
journal = {Information and Computation},
volume = {197},
number = {1},
pages = {90-121},
year = {2005},
issn = {0890-5401},
doi = {https://doi.org/10.1016/j.ic.2004.04.007},
url = {https://www.sciencedirect.com/science/article/pii/S0890540105000179},
author = {Jan Chomicki and Jerzy Marcinkowski},
}

@inproceedings{Cali03,
author = {Cal\`{\i}, Andrea and Lembo, Domenico and Rosati, Riccardo},
title = {On the decidability and complexity of query answering over inconsistent and incomplete databases},
year = {2003},
isbn = {1581136706},
publisher = {Association for Computing Machinery},
address = {New York, NY, USA},
url = {https://doi.org/10.1145/773153.773179},
doi = {10.1145/773153.773179},
booktitle = {Proceedings of the Twenty-Second ACM SIGMOD-SIGACT-SIGART Symposium on Principles of Database Systems},
pages = {260–271},
numpages = {12},
location = {San Diego, California},
series = {PODS '03}
}

@article{STAWORKO20101,
title = {Consistent query answers in the presence of universal constraints},
journal = {Information Systems},
volume = {35},
number = {1},
pages = {1-22},
year = {2010},
issn = {0306-4379},
doi = {https://doi.org/10.1016/j.is.2009.03.004},
url = {https://www.sciencedirect.com/science/article/pii/S0306437909000143},
author = {Sławomir Staworko and Jan Chomicki},
}

@InProceedings{Lopatenko07,
author="Lopatenko, Andrei
and Bertossi, Leopoldo",
editor="Schwentick, Thomas
and Suciu, Dan",
title="Complexity of Consistent Query Answering in Databases Under Cardinality-Based and Incremental Repair Semantics",
booktitle="Database Theory -- ICDT 2007",
year="2006",
publisher="Springer Berlin Heidelberg",
address="Berlin, Heidelberg",
pages="179--193",
isbn="978-3-540-69270-6"
}

@InProceedings{arming16,
  author =	{Arming, Sebastian and Pichler, Reinhard and Sallinger, Emanuel},
  title =	{{Complexity of Repair Checking and Consistent Query Answering}},
  booktitle =	{19th International Conference on Database Theory (ICDT 2016)},
  pages =	{21:1--21:18},
  series =	{Leibniz International Proceedings in Informatics (LIPIcs)},
  ISBN =	{978-3-95977-002-6},
  ISSN =	{1868-8969},
  year =	{2016},
  volume =	{48},
  editor =	{Martens, Wim and Zeume, Thomas},
  publisher =	{Schloss Dagstuhl -- Leibniz-Zentrum f{\"u}r Informatik},
  address =	{Dagstuhl, Germany},
  URL =		{https://drops.dagstuhl.de/entities/document/10.4230/LIPIcs.ICDT.2016.21},
  URN =		{urn:nbn:de:0030-drops-57900},
  doi =		{10.4230/LIPIcs.ICDT.2016.21}
}

@inproceedings{10.1145/3452021.3458334,
author = {Koutris, Paraschos and Ouyang, Xiating and Wijsen, Jef},
title = {Consistent Query Answering for Primary Keys on Path Queries},
year = {2021},
isbn = {9781450383813},
publisher = {Association for Computing Machinery},
address = {New York, NY, USA},
url = {https://doi.org/10.1145/3452021.3458334},
doi = {10.1145/3452021.3458334},
booktitle = {Proceedings of the 40th ACM SIGMOD-SIGACT-SIGAI Symposium on Principles of Database Systems},
pages = {215--232},
numpages = {18},
location = {Virtual Event, China},
series = {PODS'21}
}


\newpage
\newpage
\section{Technical Appendix}

In the following, we provide the proofs omitted in the main text. In each case, the claim of
the theorem or lemma is stated once more for the reader’s convenience.

\subsection{Proofs for Section~\ref{section:parameterized}}

\begin{claim*}[of Lemma~\ref{lemma:p-basic-aea}]
  For each modification operation $\otimes\in \{\mathrm{del}, \mathrm{add}, \mathrm{edit}\}$, $\PEM{basic}{\otimes}(aea)$ contains a $\Class{W}[2]$-hard problem.
\end{claim*}

\begin{proof}

  \emph{The formula.} Consider the formula
  \begin{align*}
    \phi_{\ref{lemma:p-basic-aea}} = \forall x \exists y \forall z
    \bigl(x\adj y \land \neg(x\adj z \land y \adj z)\bigr), %\label{lemma:p-basic-aea-eq}
  \end{align*}
  which expresses the property that every vertex has a neighbor with whom it is not part of a
  triangle. We will call such a neighbor $y$ a \emph{witness} for $x$.

  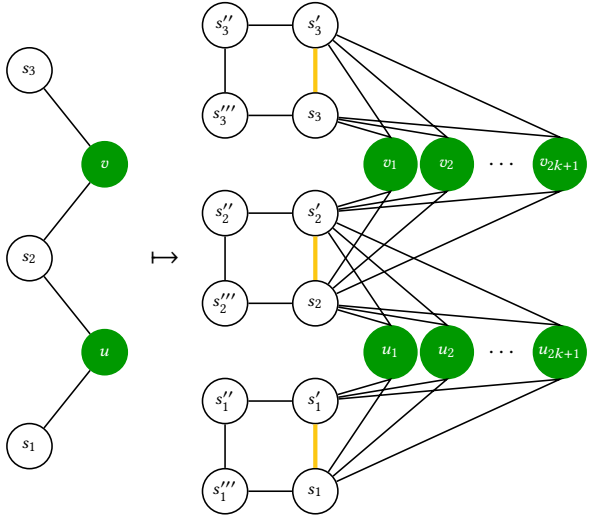
\begin{figure}
    \centering
    \begin{tikzpicture}[
        right part/.style={xshift=3.8cm},y=2.5cm,
        node/.append style={minimum size=6mm, inner sep=0pt},
        large/.style={minimum size=7mm},
      ]
      
      \foreach \i in {1,2,3} {
        \node (s\i) at (0,\i) [node] {$s_\i$};
  
        \node (ns\i)    at (0,\i) [node]      [right part,shift={( 0mm,-6mm)}] {$s_\i$};
        \node (ns\i')   at (0,\i) [node] [right part,shift={( 0mm,6mm)}] {$s_\i'$};
        \node (ns\i'')  at (0,\i) [node] [right part,shift={(-12mm,6mm)}] {$s_\i''$};
        \node (ns\i''') at (0,\i) [node] [right part,shift={(-12mm,-6mm)}] {$s_\i'''$};
  
        \draw [delete edge] (ns\i) -- (ns\i');
        \draw  (ns\i') -- (ns\i'');
        \draw  (ns\i'') -- (ns\i''');
        \draw  (ns\i''') -- (ns\i);
      }

      \foreach \i/\name in {1/u,2/v} {
        \node (\name)   at (0,\i) [green node,shift={(1,.5)}] {$\name$};
  
        \node (\name1)  at (0,\i) [green node, large]  [right part,shift={( 1,.5)}] {$\name_1$};
        \node (\name2)  at (0,\i) [green node, large]  [right part,shift={( 1.75,.5)}] {$\name_2$};
        \node (\name3)  at (0,\i) []  [right part,shift={( 2.5,.5)}] {$\dots$};
        \node (\name4)  at (0,\i) [green node, large]  [right part,shift={( 3.25,.5)}] {$\name_{2k+1}$};
      }
  
      \foreach \s/\u/\anc in {1/u/south,2/u/north,2/v/south,3/v/north} {
        \draw (s\s) -- (\u);
  
        \draw 
        (ns\s) -- (\u1.\anc) -- (ns\s')
        (ns\s) -- (\u2.\anc) -- (ns\s')
        (ns\s) -- (\u4.\anc) -- (ns\s');
      }
  
      % arrow
      \node (arrow) at (1.8,2) {\Large$\mapsto$};
    \end{tikzpicture}
    \caption{Visualization of the reduction in Lemma~\ref{lemma:p-basic-aea}. }
    \label{figure:p-basic-aea}
  \end{figure}

  \emph{The reduction.} We will reduce to the edge deletion problem. Let $(S
  \mathbin{\dot\cup} U, E, k)$ be given as
  input.  The reduction outputs $k' = k$ together with the basic
  graph $G' = (V', E')$ constructed as follows:
  \begin{itemize}
  \item For each $s \in S$, add $s$ to $V'$ and also three more
    vertices $s'$, $s''$,~$s'''$ and connect them in a cycle, that is $s \undiradj' s'
    \undiradj' s'' \undiradj'  s''' \undiradj' s$.
  \item For each $u \in U$, add $2k+1$ copies $u_1,\dots, u_{2k+1}$
    to~$V'$.  
  \item Whenever $u \undiradj s$ holds, let all $u_i$ form a triangle with $s$
    and~$s'$ in the new graph, that is, for $i\in\{1,\dots, 2k+1\}$ let
    $u_i \undiradj' s$ and $u_i \undiradj' s'$.
  \end{itemize}
  An example for the reduction is shown in Figure~\ref{figure:p-basic-aea}, using
  the same conventions as in Figure~\ref{figure:p-undir-ae}. Each $s\in S$ is made part
  of a length-4 cycle, while for each element of $U$ exactly $2k+1$ copies are
  added to the new graph. Each edge $s \undiradj u$ gets replaced by $4k+2$ edges,
  namely $u_i \undiradj' s$ and $u_i \undiradj' s'$ for all copies $u_i$ of~$u$. The
  size-$1$ set cover $C = \{s_2\}$ corresponds to the fact that deleting
  exactly $s_2 \undiradj' s_2'$ from the right graph yields a graph in which each
  vertex has a neighbor with whom it is not part of a triangle. The same is
  true for the size-$2$ set cover $C = \{s_1, s_3\}$. In contrast, $C =
  \{s_1\}$ is not a set cover as $v$ is not covered and, indeed, all four
  neighbors of~$v_1$ (namely $s_2$, $s_2'$, $s_3$, and $s_3'$) are part of
  triangles, if we  delete neither $s_2\undiradj' s_2'$ nor $s_3\undiradj' s_3'$.

  \emph{Forward direction.}
  Let  $(S \mathbin{\dot\cup} U, E, k) \in \PLang{Set-Cover}$ be given. We need
  to show that $(G', k') \in \PEM{basic}{\mathrm{del}}(\phi_{\ref{lemma:p-basic-aea}})$
  holds. Let $C \subseteq S$ with $|C| \le k$ be a cover of~$U$. We claim that $(V', E'
  \setminus C') \models \phi_{\ref{lemma:p-basic-aea}}$ for $C' = \{(s, s') \mid
  s \in C\} \cup \{(s', s) \mid s \in C\}$, that is, removing all $s \undiradj' s'$ from~$E'$
  yields witnesses satisfying $\phi_{\ref{lemma:p-basic-aea}}$.
  Let us check that this is true for all vertices in $V'$: For all $s \in S$, each of the
  vertices $s,s',s'',s'''$ has a neighbor in~$G'$ that is not part of a
  triangle, that is, a witness. This still holds for $s, s',s'',s'''$ in $(V', E'\setminus C')$, because
 all of them have witnesses among each other.

  For each $u \in U$ and $i \in \{1, \dots, 2k+1\}$, there is a neighbor of $u_i$ in $(V',
  E'\setminus C')$ that is not part of a triangle in $(V', E'\setminus C')$:
  Since $C$ is a cover, there must be an $s \in C$ with $u \undiradj' s$. But
  then, there is no other vertex with which $u$ and $s$ form a triangle, since
  the only common neighbor of $u$ and $s$ was $s'$, and
  $s\undiradj s'$ was deleted from $E'$.

  \emph{Backward direction.}
  Conversely, suppose that for $G' = (V', E')$ we are given a set $D \subseteq
  E'$ with $\|D\| \le k$ such that $(V', E'\setminus D) \models
  \phi_{\ref{lemma:p-basic-aea}}$. For each $u \in U$, consider the copies $u_i$
  for $i\in \{1,\dots,2k+1\}$. For every $s \in S$ with $s \undiradj u$ there is a
  triangle $u_i \undiradj' s \undiradj' s' \undiradj' u_i$, but every neighbor
  of~$u_i$ in $G'$ is an $s$ or $s'$ of such a triangle. Hence, in order to ensure that $\phi_{\ref{lemma:p-basic-aea}}$
  holds, we have to (1) delete either~$u_i\undiradj' s$ or $u_i\undiradj' s'$ or (2)
  delete $s \undiradj' s'$. Since there are $2k+1$ copies of~$u$, we cannot use
  option~(1) for all copies of~$u$, so for each $u \in U$ there must be an
  $s \in S$ with $s\undiradj u$ such that $s\undiradj' s'$ is deleted from $G'$. However,
  this means that the set $C = \{s \in S \mid (s, s') \in D\}$ is a set cover
  of $(S \mathbin{\dot\cup} U, E)$ and, clearly, $|C| \le \|D\| \le k$.

  \emph{Edge editing.} Adding edges does not impact the
  reduction: We call a vertex happy if it has a neighbor with whom it is
  not part of a triangle. Naturally,
  $G\models\phi_{\ref{lemma:p-basic-aea}}$ if every vertex is happy. Adding an
  edge between $u_i$ and $u_j$ with $i \neq j$ that stem from the same $u \in U$ does not make $u_i$ and $u_j$ happy, since for all $s\in
  S$ with $s\undiradj u$, we have a triangle $s\undiradj' u_i\undiradj'
  u_j\undiradj' s$. Adding
  an edge between $u_i$ and $v_j$ that are copies of distinct vertices in $U$ only makes $u_i$ and $v_j$ happy if there is no $s$
  with $s\undiradj' u_i$ and $s\undiradj' v_j$. In either case, this operation only makes up to two vertices
  happy, which is not enough to cover the $4k + 2$ vertices $u_i$ and $v_j$ with
  $k$ edge additions.

  \emph{Adaption to the directed setting.}
  We adapt the reduction above by choosing the formula
  \begin{align*}
    \phi_{\ref{lemma:p-basic-aea}, \text{dir}} = \forall x \exists y \forall z
    \bigl(x\adj y \land \neg(x\adj z \land y \adj z) \land \phi_{\text{undir}}(x, z)\bigr), %\label{lemma:p-basic-aea-eq}
  \end{align*}
  
  where $\phi_{\text{undir}}(x, y) = x\adj y \to y\adj x$, and setting the
  modification budget to $2k$. Since we defined undirected graphs as directed graphs with a symmetric edge relation,
  $\phi_{\text{undir}}$ enforces that each edit of an edge entails an edit of the edge oriented in the opposite way.
\end{proof}

\begin{claim*}[of Lemma~\ref{lemma:p-basic-aee}]
  For each modification operation $\otimes\in \{\mathrm{del}, \mathrm{add}, \mathrm{edit}\}$, $\PEM{basic}{\otimes}(aee)$ contains a $\Class{W}[2]$-hard problem.
\end{claim*}

\begin{proof}

  \emph{The formula.} Consider the formula
  \begin{align*}
    \phi_{\ref{lemma:p-basic-aee}} = \forall x \exists y \exists y'
    (x\adj y \land x\adj y' \land y\adj y'), %\label{lemma:p-basic-aee-eq}
  \end{align*}
  which expresses that every vertex is part of a triangle.

  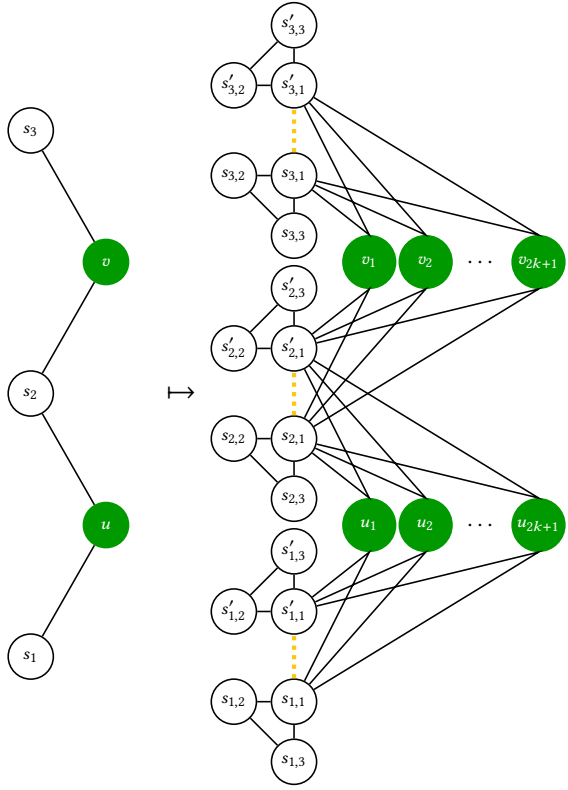
\begin{figure}
    \centering
    \begin{tikzpicture}[
        right part/.style={xshift=3.5cm},y=3.5cm,
        node/.append style={minimum size=6mm, inner sep=0pt},
        large/.style={minimum size=7mm},
      ]
      
      \foreach \i in {1,2,3} {
        \node (s\i) at (0,\i) [node] {$s_\i$};
  
        \node (ns\i)    at (0,\i) [node]      [right part,shift={( 0mm,-6mm)}] {$s_{\i, 1}$};
        \node (ns\i')   at (0,\i) [node] [right part,shift={( 0mm,6mm)}]
        {$s_{\i, 1}'$};
        
        \node (nsc1\i')  at (0,\i) [node] [right part,shift={(-8mm,6mm)}]
        {$s_{\i, 2}'$};
        \node (nsc2\i')  at (0,\i) [node] [right part,shift={(0mm,14mm)}]
        {$s_{\i, 3}'$};
        
        \node (nsc1\i) at (0,\i) [node] [right part,shift={(-8mm,-6mm)}]
        {$s_{\i, 2}$};
        \node (nsc2\i)  at (0,\i) [node] [right part,shift={(0mm,-14mm)}]
        {$s_{\i, 3}$};
  
        \draw [add edge] (ns\i) -- (ns\i');
        \draw  (ns\i') -- (nsc1\i') -- (nsc2\i') -- (ns\i');
        \draw  (ns\i) -- (nsc1\i) -- (nsc2\i) -- (ns\i);
      }

      \foreach \i/\name in {1/u,2/v} {
        \node (\name)   at (0,\i) [green node,shift={(1,.5)}] {$\name$};
  
        \node (\name1)  at (0,\i) [green node, large]  [right part,shift={( 1,.5)}] {$\name_1$};
        \node (\name2)  at (0,\i) [green node, large]  [right part,shift={( 1.75,.5)}] {$\name_2$};
        \node (\name3)  at (0,\i) []  [right part,shift={( 2.5,.5)}] {$\dots$};
        \node (\name4)  at (0,\i) [green node, large]  [right part,shift={( 3.25,.5)}] {$\name_{2k+1}$};
      }
  
      \foreach \s/\u/\anc in {1/u/south,2/u/north,2/v/south,3/v/north} {
        \draw (s\s) -- (\u);
  
        \draw 
        (ns\s) -- (\u1.\anc) -- (ns\s')
        (ns\s) -- (\u2.\anc) -- (ns\s')
        (ns\s) -- (\u4.\anc) -- (ns\s');
      }
  
      % arrow
      \node (arrow) at (2,2) {\Large$\mapsto$};
    \end{tikzpicture}
    \caption{Visualization of the reduction in Lemma~\ref{lemma:p-basic-aee}. }
    \label{figure:p-basic-aee}
  \end{figure}

  \emph{The reduction.} We will reduce to the edge addition problem. Let $(S
  \mathbin{\dot\cup} U, E, k)$ be given as
  input.  The reduction outputs $k' = k$ and a basic
  graph $G' = (V', E')$ constructed as follows:
  \begin{itemize}
  \item For each $s_i \in S$, add vertices named $s_{i, j}, s_{i, j}'$ for $j\in
  \{1, 2, 3\}$ to $V'$ and add edges
  $s_{i, 1} \undiradj' s_{i, 2} \undiradj' s_{i, 3}\undiradj' s_{i, 1}$ as well as
  $s_{i, 1}' \undiradj' s_{i, 2}' \undiradj' s_{i, 3}'\undiradj' s_{i, 1}'$.
  \item For each $u \in U$, add $2k+1$ copies $u_1,\dots, u_{2k+1}$ 
    to~$V'$.  
  \item Whenever $u \undiradj s$ holds, connect all $u_i$ with $s_{i, 1}$
    and~$s_{i, 1}'$, that is, for $i\in\{1,\dots, 2k+1\}$ let
    $u_i \undiradj' s_{i, 1}$ and $u_i \undiradj' s_{i, 1}'$. 
  \end{itemize}
  An example for the reduction is shown in Figure~\ref{figure:p-basic-aee},
    using the same conventions as in the previous figures. For each~$s\in S$, we
    add two triangles, and for each element of $U$ exactly $2k+1$ copies. Each
    edge $s_j \undiradj u$ gets replaced by $4k+2$ edges, namely $u_i \undiradj'
    s_{j, 1}$ and $u_i \undiradj' s_{j, 1}'$ for all copies $u_i$ of~$u$. The
    dotted yellow edges are not present, but candidates for edge addition. The
    size-$1$ set cover $C = \{s_2\}$ corresponds to the fact that adding the
    edge $s_{2, 1} \undiradj' s_{2, 1}'$ to the right graph yields a graph in
    which each vertex is part of a triangle. The same holds for the size-$2$ set
    cover $C = \{s_1, s_3\}$. In contrast, $C = \{s_1\}$ is not a set cover, as
    $v$ is not covered and, indeed, $v_1$ is not part of any triangle, if we do
    not add $s_{2, 1}\undiradj' s_{2, 1}'$ or $s_{3, 1}\undiradj' s_{3, 1}'$.

  \emph{Forward direction.}
  Let  $(S \mathbin{\dot\cup} U, E, k) \in \PLang{Set-Cover}$ be given. We show that $(G', k') \in \PEM{basic}{\mathrm{add}}(\phi_{\ref{lemma:p-basic-aee}})$
  holds. Let $C \subseteq S$ with $|C| \le k$ be a cover of~$U$ and let $C' = \{(s_{i, 1}, s_{i, 1}')
  \mid s_{i, 1}\in C\} \cup \{(s_{i, 1}', s_{i, 1})
  \mid s_{i, 1}\in C\}$. We claim that $(V', E'
  \cup C') \models \phi_{\ref{lemma:p-basic-aee}}$, that is, adding all $s_{i, 1}
  \undiradj' s_{i, 1}'$ for $s_{i, 1}\in C$ to~$E'$
  yields a graph in which every vertex is part of a triangle.
  Let us check that this is true for all vertices in~$V'$: For all $s_i \in S$, both $s_{i, j}$ for $j\in\{1, 2, 3\}$, and $s_{i, j}'$ for $j\in\{1, 2,
  3\}$ constitute triangles. For each $u \in U$ and $i \in \{1, \dots, 2k+1\}$, $u_i$ is part of a triangle in $(V',
  E'\cup C')$: Since $C$ is a cover, there must be an $s_{i, 1}
  \in C$ with $u \undiradj s_{i, 1}$. Then, we also have $u_i\undiradj' s_{i, 1} \undiradj' s_{i, 1}' \undiradj' u_i$ by
  construction.

  \emph{Backward direction.}
  Suppose that for $G' = (V', E')$ we are given a set $A \subseteq
  V'\times V'$ with $\|A\| \le k$ such that $(V', E'\cup A) \models
  \phi_{\ref{lemma:p-basic-aee}}$. Given $A$, if necessary, we modify it as
  follows: If $(s, t') \in A$ for two distinct $s, t\in S$, remove $(s, t')$
  and $(t', s)$ from $A$ and add $(s, s')$ and $(s', s)$ instead. Clearly,
  this does not increase the size of $A$. Crucially, if $(s, t')$ used to be part
  of a witness triangle for some $u_i$, there is now a new witness triangle $u_i \undiradj' s\undiradj' s' \undiradj'
  u_i$.
  
  Now, let $C$ contain all $s\in S$ such that $(s, s')\in A$ holds. Clearly,
  $|C| \leq \|A\| \leq k$. We claim that $C$ is a cover: Let $u\in U$ be given.
  Each of the $2k + 1$ many $u_i$ are part of a triangle in $E'\cup A$, but this is not the case for any of them in $E'$.
  Consequently, for each~$u_i$, we have an
  edge in~$A$ that is either incident to~$u_i$ or that connects two
  neighbors of~$u_i$. Since $\|A\|\leq k$, it is impossible that $A$
  contains an incident 
  edge for each of the $2k + 1$ many $u_i$. Thus, $A$ connects two neighbors
  of~$u_i$, and the only candidates are $s$ and~$s'$ for $s\in S$.
  Since we ensured that the only edges in $A$ between $s$ and~$t'$ vertices have
  $s = t$, we get $s\in C$. Thus, $u$ is covered by~$C$.

  \emph{Edge editing.}
  To see that the reduction works for edge editing, note that deleting edges does not create new triangles.

  \emph{Adaption to the directed setting.} We choose
  \begin{align*}
    \phi_{\ref{lemma:p-basic-aee}, \text{dir}} = \forall x \exists y \exists y'
    (x\adj y \land x\adj y' \land (y\adj y' \lor y'\adj y)),
  \end{align*}
  and perform the reduction in the same way. Since the formula states that the direction of the edge
  between $y$ and $y'$ is not important, it suffices to add one edge between
  $s_{i, 1}$ and $s_{i, 1}'$.
  
\end{proof}

\begin{claim*}[of Lemma~\ref{lemma:p-basic-aae}]
  For each modification operation $\otimes\in \{\mathrm{del}, \mathrm{add}, \mathrm{edit}\}$, $\PEM{basic}{\otimes}(aae)$ contains a $\Class{W}[2]$-hard problem.
\end{claim*}

\begin{proof}
  The statement follows directly from the fact that the problem $\Lang{Edge-Adding to
  Diameter }2$ is $\Class{W}[2]$-hard~\cite{GaoHN13}. The problem can be expressed
  as $\PEM{basic}{\mathrm{add}}(\phi_{\text{diam-$2$}})$ with
  \begin{align*}
    \phi_{\text{diam-$2$}} = \forall x \forall x' \exists y(x\nadj x' \to (x \adj y\land x'\adj y)).
  \end{align*}
  
  Here, we give an alternative proof for the statement by showing
  $\Class{W}[2]$-hardness for $\PLang{Edge-Adding to Triangle Edge
  Cover}$, a problem which is of independent interest in
  discrete mathematics~\cite{BujtasDDGTY25, ErdosGT96}, which is definable as $\PEM{basic}{\mathrm{add}}(\phi_{\text{triangle-edge-cover}})$.
  It asks whether we can add at most $k$ edges to a graph such that
  each edge is covered by a triangle.
  
  \emph{The formula.} Consider the formula
  \begin{align*}
    \phi_{\text{triangle-edge-cover}} = \forall x \forall x' \exists y
    \bigl(x\adj x' \to (x\adj y \land x'\adj y)\bigr), %\label{lemma:p-basic-aae-eq}
  \end{align*}
  which expresses that every edge is part of a triangle (observe the
  similarity to~$\phi_{\text{diam-$2$}}$).

  \begin{figure*}
    \centering
    \begin{tikzpicture}[
      right part/.style={xshift=6cm},y=2.5cm,
      node/.append style={minimum size=6mm, inner sep=0pt},
      large/.style={minimum size=7mm},
    ]
      \begin{scope}[yshift=1cm]
        \node (s1) at (0, 0) [node] {$s_1$};
        \node (s2) at (0, 2) [node] {$s_2$};
        \node (s3) at (0, 4) [node] {$s_3$};
    
        \node (u) at (1, 1) [green node] {$u$};
        \node (v) at (1, 3) [green node] {$v$};

        \draw (s1) -- (u);
        \draw (s2) -- (u);
        \draw (s2) -- (v);
        \draw (s3) -- (v);
      \end{scope}

      \begin{scope}[xshift=5cm]
        \node (s1) at (0, 0) [node] {$s_1$};
        \node (s1') at (0, 1) [node] {$s_1'$};
        \node (s2) at (0, 2) [node] {$s_2$};
        \node (s2') at (0, 3) [node] {$s_2'$};
        \node (s3) at (0, 4) [node] {$s_3$};
        \node (s3') at (0, 5) [node] {$s_3'$};

        \node (c) at (6, 2.5) [red node] {$c$};

        \node (u1) at (3, 1) [green node, large] {$u_1$};
        \node (u2) at (3, 1.5) [green node, large] {$u_2$};
        \node (u31) at (3, 1.725) [circle, fill=black, inner sep=0pt, minimum size=1pt] {};
        \node (u32) at (3, 1.775) [circle, fill=black, inner sep=0pt, minimum size=1pt] {};
        \node (u33) at (3, 1.825) [circle, fill=black, inner sep=0pt, minimum size=1pt] {};
        \node (u4) at (3, 2.05) [green node] {$u_{2k + 1}$};

       % \node (v) at (4, 3) [green node] {$v$};

        \node (v1) at (3, 3) [green node, large] {$v_1$};
        \node (v2) at (3, 3.5) [green node, large] {$v_2$};
        \node (v31) at (3, 3.725) [circle, fill=black, inner sep=0pt, minimum size=1pt] {};
        \node (v32) at (3, 3.775) [circle, fill=black, inner sep=0pt, minimum size=1pt] {};
        \node (v33) at (3, 3.825) [circle, fill=black, inner sep=0pt, minimum size=1pt] {};
        \node (v4) at (3, 4.05) [green node] {$v_{2k + 1}$};

        \draw  (u1) -- (c);
        \draw  (u2) -- (c);
        \draw  (u4) -- (c);

        \draw  (v1) -- (c);
        \draw  (v2) -- (c);
        \draw  (v4) -- (c);

        \draw  (s1) -- (s1');
        \draw  (s2) -- (s2');
        \draw  (s3) -- (s3');

        \draw  (s1) -- (u1);
        \draw  (s1) -- (u2);
        \draw  (s1) -- (u4);
        \draw  (s1') -- (u1);
        \draw  (s1') -- (u2);
        \draw  (s1') -- (u4);

        \draw  (s2) -- (u1);
        \draw  (s2) -- (u2);
        \draw  (s2) -- (u4);
        \draw  (s2') -- (u1);
        \draw  (s2') -- (u2);
        \draw  (s2') -- (u4);

        \draw  (s2) -- (v1);
        \draw  (s2) -- (v2);
        \draw  (s2) -- (v4);
        \draw  (s2') -- (v1);
        \draw  (s2') -- (v2);
        \draw  (s2') -- (v4);

        \draw  (s3) -- (v1);
        \draw  (s3) -- (v2);
        \draw  (s3) -- (v4);
        \draw  (s3') -- (v1);
        \draw  (s3') -- (v2);
        \draw  (s3') -- (v4);

        \draw [add edge] (s1) -- (c);
        \draw [add edge] (s2) -- (c);
        \draw [add edge] (s3) -- (c);
      \end{scope}

      % arrow
      \node (arrow) at (3,2.5) {\Large$\mapsto$};
    \end{tikzpicture}
    \caption{Visualization of the reduction in Lemma~\ref{lemma:p-basic-aae}.
    }
    \label{figure:p-basic-aae}
  \end{figure*}
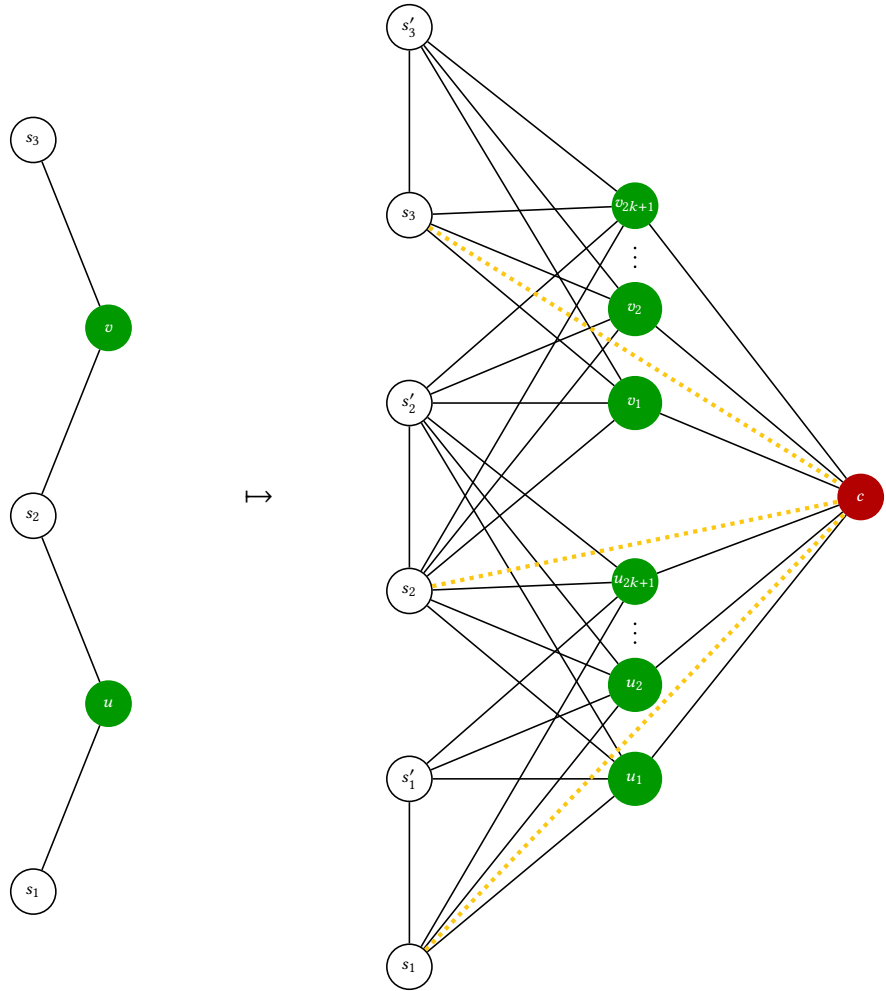

  \emph{The reduction.} We will reduce to the edge addition problem. Let $(S
  \mathbin{\dot\cup} U, E, k)$ be given as
  input.  The reduction outputs $k' = k$ and a basic
  graph $G' = (V', E')$ constructed as follows:
  \begin{itemize}
  \item For each $s_i \in S$, add vertices named $s_{i}, s_{i}'$ to $V'$ and connect them by an edge.
  \item Add a vertex $c$ to $V'$.
  \item For each $u \in U$, add $2k+1$ copies $u_1,\dots,
  u_{2k+1}$ and connect them to $c$.
  \item Whenever $u \undiradj s$ holds, connect all $u_i$ to $s$
    and~$s'$, that is, for $i\in\{1,\dots, 2k+1\}$, let
    $u_i \undiradj' s$ and $u_i \undiradj' s'$. 
  \end{itemize}
  An example for the reduction is shown in Figure~\ref{figure:p-basic-aae}, using
    the same conventions as in the previous figures. Added gadgets
    that depend neither on $S$ nor on $U$ are highlighted in red. For each $s\in S$, we add
    vertices $s, s'$ with $s\undiradj' s'$, and for each element of $U$ we add exactly $k+1$ copies. Each edge $s \undiradj u$ gets replaced by edges connecting
    both $s$ and $s'$ to every copy $u_i$ of~$u$. Finally, a vertex $c$ is added
    and connected to every copy $u_i$ of $u$ for all $u\in U$. The size-$1$ set cover $C = \{s_2\}$ corresponds to the fact that
    adding the edge $c \undiradj' s_2$ yields a graph
    in which each edge is part of a triangle: The edges $s\undiradj' s'$, $s\undiradj' u_i$,
    and $s'\undiradj' u_i$ already form triangles, and now, the edges $c\undiradj' u_i$ are
    part of a triangle with $s\undiradj' u_i$ and $s\undiradj' c$. The same is true for the
    size-$2$ set cover $C = \{s_1, s_3\}$. In contrast, $C = \{s_1\}$ is not a
    set cover, as $v$ is not covered and, indeed, $v_1\undiradj' c$ is not part of any
    triangle, if we add neither $s_2\undiradj' c$ nor $s_3\undiradj' c$.

  \emph{Forward direction.}
  Let  $(S \mathbin{\dot\cup} U, E, k) \in \PLang{Set-Cover}$ be given. We show that $(G', k') \in
  \PEM{basic}{\mathrm{add}}(\phi_{\text{triangle-edge-cover}})$ holds. Let $C
  \subseteq S$ with $|C| \le k$ be a cover of~$U$. Let $C' = \{(s, c) \mid s\in
  C\}\cup \{(c, s) \mid s\in
  C\}$. We claim that $(V', E' \cup C') \models \phi_{\text{triangle-edge-cover}}$,
  that is, adding all $s \undiradj' c$ for $s\in C$ to~$E'$ yields a graph in which every edge is part
  of a triangle. Let us check all edges in $G'$: For edges between
  $s_i$ and $u_j$, the edges $s_i \undiradj' u_j$, $s_i' \undiradj' u_j$, and $s_i \undiradj'
  s_i'$ form a triangle.
  Moreover, each of the copies
  $u_i$ of~$u$ is connected to~$c$, and since $C$ is a cover, by construction, we
  have a triangle consisting of the edges $s \undiradj' u_i$, $s \undiradj' c$, and $u_i
  \undiradj' c$ for $i\in \{1, \dots, 2k+1\}$.

  \emph{Backward direction.}
  Conversely, suppose that for $G' = (V', E')$ we are given a set $A \subseteq
  V' \times V'$ with $\|A\| \le k$ such that $(V', E'\cup A) \models
  \phi_{\text{triangle-edge-cover}}$. If necessary, we modify~$A$ as
  follows: If $(c, s') \in A$, remove $(c, s')$
  and also $(s', c)$ from $A$ and add $(c, s)$ and $(s, c)$ instead. Clearly,
  this cannot increase~$\|A\|$. Crucially, if $(c, s')$ used to be part
  of a witness triangle for some~$u_i$, there is now a witness
  triangle for~$u_i$ that uses $s\undiradj' c$,
  $s\undiradj' u_i$, and $s'\undiradj'
  u_i$.
  
  Now, let $C$ contain all $s\in S$ such that $(c, s)\in A$ holds. Clearly, $|C|
  \leq \|A\| \leq k$. We claim that $C$ is a cover: Let $u\in U$ be given. Each
  of the $2k + 1$ many edges $u_i\undiradj' c$ must be part of a triangle in
  $E'\cup A$. Since none of them are in a triangle in~$E'$, for each $u_i$ we
  must have an edge in $A$ that is either incident to $u_i$ or that connects $c$
  to a neighbor of $u_i$. Since $\|A\|\leq k$, it is impossible that $A$
  contains an incident edge for each of the $2k + 1$ many $u_i$. Thus, $A$ 
  connects $c$ to a neighbor of $u_i$. However, its only neighbors are $s$
  and $s'$ for $s\in S$. Since we ensured that there are no edges between
  $c$ and $s'$ in $A$, we get $s\in C$. Thus, $u$ is covered by $C$.

  \emph{Edge editing.}
  To see that the reduction also works for edge editing, observe that
  for every vertex~$u$ we have $2k+1$ edges $u_i \undiradj' c$ that
  are not part of a triangle. Since we 
  can remove at most $k$ of them, we still have $k+1$ edges $u_i
  \undiradj' c$ left that need to be made part of triangles. This forces us to \emph{add} an edge $c \undiradj' s$ such that
  $s \undiradj u$.

  \emph{Adaption to the directed setting.}
  Similarly to Lemma~\ref{lemma:p-basic-aea}, we can adapt the reduction by choosing the formula
  \begin{align*}
    \phi_{\text{tec, dir}} = \forall x \forall x' \exists y
    \bigl((x\adj x' \to (x\adj y \land x'\adj y)) \land \phi_{\text{undir}}(x, y)\bigr), %\label{lemma:p-basic-aea-eq}
  \end{align*}
  
  where $\phi_{\text{undir}}(x, y) = x\adj y \to y\adj x$, and setting the
  modification budget to $2k$.
\end{proof}

\subsection{Proofs for Section~\ref{section:classical}}

\begin{claim*}[of Lemma~\ref{lemma:basic-ae-aa-ea}]
  For each $\otimes\in \{\mathrm{del}, \mathrm{add}, \mathrm{edit}\}$, we have 
  $\EM{undir}{\otimes}(ae)$, $\EM{undir}{\otimes}(ea)$, $\EM{undir}{\otimes}(aa) \not\subseteq \Class{AC^0}$.
\end{claim*}

\begin{proof}
  For every prefix, we reduce from the $\Lang{Majority}$ problem.

  For $p = ae$, consider the formula $\phi_1 = \forall x \exists y (x\adj y)$. We
  reduce from $\Lang{Majority}$ to $\EM{undir}{\mathrm{add}}(\phi_1)$.
  From the bitstring $s$, we construct a basic graph $G = (V, E)$ by adding two
  isolated vertices for every bit that is set to $0$, and setting $k = |s|/2$. Now, clearly, we can make
  the graph satisfy the formula if and only if at least half of the bits in $s$
  were set to $1$.

  For $p = ea$, consider the formula $\phi_2 = \exists x \forall y (x\adj y)$. We
  reduce from $\Lang{Majority}$ to $\EM{undir}{\mathrm{add}}(\phi_2)$.
  From the bitstring $s$, we construct a basic graph $G = (V, E)$ by adding a
  vertex for every bit. Second, we add another vertex $c$ that is connected to all
  vertices that represent bits set to $1$. Again, we let $k = |s|/2$. To satisfy
  the formula, we must connect all isolated vertices to $c$, which is possible
  if and only if at least half of the bits in $s$
  were set to $1$.

  For $p = aa$, consider the formula $\phi_3 = \forall x \forall y (x\nadj y)$.  We
  reduce from $\Lang{Majority}$ to $\EM{undir}{\mathrm{add}}(\phi_3)$.
  From the bitstring $s$, we construct a basic graph $G = (V, E)$ by adding two
  vertices $v, u$ with $v\undiradj u$ for every bit in $s$ that is set to $0$ and setting $k = |s|/2$.
\end{proof}

\end{document}